\documentclass[twoside,a4paper]{article}

\usepackage{amsmath,graphicx,amssymb,fancyhdr,amsthm,enumerate,textcomp,amscd} 
\usepackage[usenames]{color}
\usepackage[all]{xy}
\usepackage{enumitem}
\newtheorem{thm}{Theorem}[section] 
\newtheorem{cor}[thm]{Corollary} 
\newtheorem{lem}[thm]{Lemma} 
\newtheorem{prop}[thm]{Proposition}
 
\newtheorem{ex}[thm]{Example}
\theoremstyle{definition} 
\newtheorem{defn}[thm]{Definition}
  
\theoremstyle{remark}  
\newtheorem{rem}[thm]{Remark}  
\def\beq{\begin{eqnarray}}  
\def\eeq{\end{eqnarray}}  
\def\bsp{\begin{split}}  
\def\esp{\end{split}}

\def\R{ \mathbb{R}}
  
\def\A{{\sf A}}  
\def\C{{\mathbb{C}}}

\newcommand{\mb}[1]{{\mathbb #1}}

\begin{document}   
   
\title{\Large\textbf{Real GIT with applications to compatible representations and Wick-rotations}}  
\author{{\large\textbf{Christer Helleland} and \textbf{Sigbj\o rn Hervik }    }
 \vspace{0.3cm} \\     
Faculty of Science and Technology,\\     
 University of Stavanger,\\  N-4036 Stavanger, Norway         
\vspace{0.3cm} \\      
\texttt{ christer.helleland@uis.no}, \\
\texttt{sigbjorn.hervik@uis.no} }     
\date{\today}     
\maketitle   
\pagestyle{fancy}   
\fancyhead{} 
\fancyhead[EC]{Helleland and Hervik}   
\fancyhead[EL,OR]{\thepage}   
\fancyhead[OC]{Real GIT: compatible representations and Wick-rotations}   
\fancyfoot{} 

\begin{abstract} 
Motivated by Wick-rotations of pseudo-Riemannian manifolds, we study real geometric invariant theory (GIT) and compatible representations. We extend some of the results from earlier works \cite{W2,W1}, in particular,
we give some sufficient as well as necessary conditions for when pseudo-Riemannian manifolds are Wick-rotatable to other signatures.  For arbitrary signatures, we consider a Wick-rotatable pseudo-Riemannian manifold with closed $O(p,q)$-orbits, and thus generalise the existence condition found in \cite{W1}. Using these existence conditions we also derive an invariance theorem for Wick-rotations of arbitrary signatures. 

\end{abstract}

\section{Introduction}
 Let $(M,g)$ be a real analytic pseudo-Riemannian manifold. Here we will ask the question: When can such a manifold be Wick-rotated to a (different) pseudo-Riemannian manifold? 

A partial answer to this question has already been given in the special case where $(M,g)$ (of arbitrary signature) is Wick-rotated to a Riemannian space at a fixed point $p$, implying that $(M,g)$ would have to be Riemann purely electric (RPE), see \cite{W1}. Standard examples of Wick-rotations can be found within Lie groups, indeed any two semi-simple real forms: $G\subset G^{\mb{C}}\supset \tilde{G}$, of a complex Lie group are Wick-rotated, where the Lie groups are equipped with their left-invariant Killing forms: $-\kappa(\cdot,\cdot)$ respectively. As explored in \cite{W2}, the existence of a Wick-rotation at a fixed point $p$ implies the existence of a Wick-rotation of the isometry groups of the pseudo-inner products on the tangent spaces at $p$: $O(p,q)\subset O(n,\mb{C})\supset O(\tilde{p},\tilde{q})$ at the identity element. We continue this study by using results of real GIT applied to actions of these groups. The results are then applied to Wick-rotations, and we give partial answers to the question above in the case of arbitrary signatures (not necessarily Riemannian). 

Another motivation behind studying such Wick-rotations are considering pseudo-Riemannian spaces having identical polynomial curvature invariants \cite{VSI,OP,HC,GW}. Consider two pseudo-Riemannian manifolds $(M,g)$ and $(\tilde{M},\tilde{g})$. Assume that all of their polynomial curvature invariants are identical, what can we then say about the relation between the two spaces? Indeed, here we will address this question locally and we reach a partial classification of spaces with identical invariants. Indeed, again, the Wick-rotations play an important role in this classification. 

Our paper is organised as follows. We begin by the study of real GIT, and apply the results to \emph{compatible representations}, which are defined and purely motivated by the study of Wick-rotations in \cite{W2,W1}.  Many of these results obtained are generalisations of previous results \cite{RS,W2,W1,PV}. These results are then applied to pseudo-Riemannian manifolds and holomorphic Riemannian manifolds. The main GIT results of our paper is Section 5, which we apply to the setting of Wick-rotations (Section 6).

\

In this paper we will reserve the notion of \emph{Riemannian space} to the case when the metric is positive definite (of signature $(++..+)$) while a \emph{Lorentzian space} has signature $(-++..+)$. Note also that the existence of the "anti-isometry" which switches the sign of the metric, $g\mapsto -g$ which induces the group isomorphism $O(p,q)\rightarrow O(q,p)$. 
   
\section{Mathematical Preliminaries}

\subsection{Real slices and compatibility} 

\begin{defn}
A \emph{holomorphic inner product space} is a complex vector space $E$ equipped with a non-degenerate complex bilinear form $g$. 
\end{defn}
For a holomorphic inner product space $E$ we can always choose an orthonormal basis. By doing so we can identify $E$ with $\C^n$ and the holomorphic inner product can be written as
\beq \label{g0}
g_0(X,Y)=X_1Y_1+...+X_nY_n, 
\eeq
where $X=(X_1,...,X_n)$ and $Y=(Y_1,...,Y_n)$. 

Using this orthonormal basis it is also convenient to consider the group of transformation leaving the holomorphic inner product invariant. Consider a complex-linear map $A: E\longrightarrow E$. Using an orthonormal basis, we can represent the map as a complex matrix $ \A: \C^n \longrightarrow \C^n$. Requiring that $g_0(A(X),A(Y))=g_0(X,Y)$, for all $X,Y$, implies that $\A^t\A={\sf 1}$. Consequently, the matrix $\A$ must be  a complex orthogonal matrix; i.e., $\A\in O(n,\C)$. 

\begin{defn}\label{realslice}
Given a holomorphic inner product space $(E,g)$. Then if $W\subset E$ is a real linear subspace for which $g\big{|}_W$ is non-degenerate and real valued, i.e., $g(X,Y)\in \R,~ \forall X,Y\in W$, we will call $W$ a \emph{real slice}. 
\end{defn}

A non-degenerate symmetric real bilinear form shall be called a pseudo-inner product.
\\

We recall that a \emph{conjugation map} $\sigma$ of a complex vector space $E$, is a real linear isomorphism: $E\xrightarrow{\sigma} E$, which is anti-linear, i.e $\sigma(ix)=-i\sigma(x)$ for all $x\in E$. The fix points of such a map, defines what is called a \emph{real form} of $E$. Thus for a complex Lie group $G$, an \emph{anti-holomorphic involution} (or \emph{real structure}): $G\xrightarrow{F} G$, is an involution of real Lie groups such that the differential at $1$: $\mathfrak{g}\xrightarrow{dF}\mathfrak{g}$, is a conjugation map.

Let $W\subset (E,g)$ be a real slice of dimension: $Dim_{\mb{R}}(W)=Dim_{\mb{C}}(E)$ (i.e $W$ is a real form of $E$). Denote $(p,q)$ for the signature of the restricted pseudo-inner product: $g\big{|}_W(-,-)$. Let $O(p,q)$ denote the real Lie group consisting of isometries of the pseudo-inner product space: $\Big{(}W, g\big{|}_W(-,-)\Big{)}$, then $O(p,q)$ is a real form of $O(n,\mb{C})$ (the isometries of $(E,g)$), by noting the anti-holomorphic involution (real structure): $A\mapsto \sigma\circ A\circ \sigma$, where $\sigma$ is the conjugation map of $W$ in $E$. 

\begin{defn} Let $W\subset (E,g)$ be a real slice. We say an involution $W\xrightarrow{\theta} W$, is a \emph{Cartan involution} of $W$, if $g_{\theta}(\cdot,\cdot):=g\big{|}_{W}(\cdot,\theta(\cdot))$, is an inner product on $W$. \end{defn}

We note that the definition generalises the notion of a Cartan involution of a semi-simple Lie algebra. 

\begin{defn} Two real forms $V$ and $\widetilde{V}$ of $E$ are said to be compatible if their conjugation maps commute, i.e $[\sigma,\tilde{\sigma}]=0$.  \end{defn}

Let $V, \tilde{V}$ and $W$ be real slices of $(E,g)$ (all of the same real dimension as $Dim_{\mathbb{C}}(E)$). Assume $g\big{|}_W(-,-)$ is an inner product, such a real slice is referred to as a \emph{compact real slice}. If all of their conjugation maps are pairwise compatible, then we shall refer to the triple: $\Big{(}V,\tilde{V}, W \Big{)}$, as a \textsl{compatible triple}. 

We shall say that $V\subset (E,g)$ is a real form, to mean that $V$ is a real slice and $Dim_{\mathbb{R}}(V)=Dim_{\mathbb{C}}(E)$.
\\

For Lie groups we define compatibility locally:

\begin{defn} Let $G\subset G^{\mathbb{C}}\supset \tilde{G}$ be two real Lie subgroups of a complex Lie group such that the real Lie algebras are real forms of $\mathfrak{g}^{\mb{C}}$. Then we say $G$ and $\tilde{G}$ are \emph{compatible} if the Lie algebras are compatible. \end{defn}

For example the abelian Lie groups: $S^1\subset \mathbb{C}^{\times}\supset \mathbb{R}^{\times}$ are compatible w.r.t to the real structures: $z\mapsto\frac{1}{\overline{z}}$ and $z\mapsto \overline{z}$ respectively. This is also an example of a compatible triple: $\Big{(}\mathbb{R}^{\times}, S^1, S^1\Big{)}$, in the sense of the following definition:

\begin{defn} Let $G\subset G^{\mathbb{C}}\supset \tilde{G}$ and $U\subset G^{\mathbb{C}}$ be real Lie subgroups of a complex Lie group such that the real Lie algebras are real forms of $\mathfrak{g}^{\mb{C}}$. Moreover assume $U$ is compact. Then we say $\Big{(}G, \tilde{G}, U\Big{)}$ is a \emph{compatible triple} if the Lie algebras are pairwise compatible. \end{defn}

\subsection{A Wick-rotation implies a standard Wick-rotation}

We recall some definitions from \cite{W2}, and prove the equivalence: $$\exists \textsf{ A Wick-rotation} \Leftrightarrow \exists \textsf{ A\ standard Wick-rotation}.$$

\begin{defn}
Given a complex manifold $M^{\mb{C}}$ with complex Riemannian metric $g^{\mb{C}}$. If a submanifold $M\subset M^{\mb{C}}$ for any point $p\in M$ we have that $T_pM$ is a real slice of $(T_pM^{\mb{C}},g^{\mb{C}})$  (in the sense of  Defn. \ref{realslice}), we will call $M$ a real slice of $(M^{\mb{C}},g^{\mb{C}})$. 
\end{defn}
This definition implies that the induced metric from $M^{\mb{C}}$ is real valued on $M$. $M$ is therefore a pseudo-Riemannian manifold. This further implies that real slices are totally real manifolds. 

\begin{defn}[Wick-related spaces]
Two pseudo-Riemannian manifolds $M$ and $\tilde{M}$ are said to be \emph{Wick-related} if there exists a holomorphic Riemannian manifold $(M^{\mb{C}},g^{\mb{C}})$ such that $M$ and $\tilde{M}$ are embedded as real slices of $M^{\mb{C}}$. 
\end{defn}
Wick-related spaces were defined in \cite{PV}. However, we also find it useful to define: 
\begin{defn}[Wick-rotation]
	If two Wick-related spaces (of the same real dimension) intersect at a point $p$ in $M^{\mb{C}}$, then we will use the term \emph{Wick-rotation}: the manifold $M$ can be Wick-rotated to the manifold $\tilde{M}$ (with respect to the point $p$).  
\end{defn}

\begin{rem}Throughout this paper, we shall always assume that $Dim_{\mathbb{R}}(M)=Dim_{\mathbb{R}}(\tilde{M})=Dim_{\mathbb{C}}(M^{\mathbb{C}})$.\end{rem}

\begin{defn}[Standard Wick-rotation]\label{sW}
	Let the $M$ and $\tilde{M}$ be Wick-related spaces (of the same dimension) having a common point $p$. Then if the tangent spaces $T_pM$ and $T_p\tilde{M}$ are embedded: $$T_pM, T_p\tilde{M}\hookrightarrow (T_pM)^{\C}\cong (T_p\tilde{M})^{\C}\hookrightarrow T_p M^{\mb{C}},$$ such that they form a compatible triple with a compact real slice $W\subset (T_pM)^{\C}\cong (T_p\tilde{M})^{\C}$, then we say that the spaces $M$ and $\tilde{M}$ are related through a \emph{standard Wick-rotation}. 
\end{defn}

It is useful to note that a standard Wick-rotation: $(M,g)\subset (M^{\mathbb{C}},g^{\mb{C}})\supset (\tilde{M},\tilde{g})$, at a common point $p$, induces a Wick-rotation of Lie groups at $1$: $O(p,q)\subset O(n,\mb{C})\supset O(\tilde{p},\tilde{q})$. This observation is for instance used in \cite{W1}, and is seen as follows. Let $\{e_1,\dots, e_p,\dots e_n\}$ be a pseudo-orthonormal basis of $g(-,-)$, and $\theta$ denote the Cartan involution of $g$ w.r.t this basis. Then $\{e_1,\dots, e_p, ie_{p+1},\dots ie_n\}:=\{y_1,\dots y_n\}$ is an orthonormal basis of $g^{\mb{C}}(-,-)$. Thus define a holomorphic inner product $\textbf{g}^{\mb{C}}$ on $End(T_pM^{\mb{C}})$ by: $$\textbf{g}^{\mb{C}}(f,h):=\sum_{1\leq l\leq n} g^{\mb{C}}(f(y_l), h(y_l)).$$ It is easy to check that $End(T_pM)\subset \Big{(}End(T_pM^{\mb{C}}), \textbf{g}^{\mb{C}} \Big{)}\supset End(T_p\tilde{M})$ are real forms, precisely because $T_pM$ and $T_p\tilde{M}$ are compatible with the compact real slice: $W:=\langle y_1,\dots, y_p, iy_{p+1}, \dots iy_n\rangle$. A natural choice of Cartan involution $\Theta$ of the induced pseudo-inner product $\textbf{g}$ on $End(T_pM)$ is given by: $$f\mapsto \theta f\theta, \ f\in End(T_pM).$$ Note that if we restrict to the pseudo-orthogonal Lie algebra $\mathfrak{o}(p,q)\subset End(T_pM)$, then $\Theta$ leaves invariant $\mathfrak{o}(p,q)$. Moreover if $p+q\geq 3$ then $\Theta$ is a Cartan involution of the semi-simple Lie algebra: $\mathfrak{o}(p,q)$. An easy calculation shows that $\textbf{g}$ is invariant under the conjugation action of $O(p,q)$ on $End(T_pM)$: $$h\cdot f:=h f h^{-1}, \ h\in O(p,q), \ f\in End(T_pM).$$ Thus $\textbf{g}$ induces a bi-invariant metric on $O(p,q)$. If $p+q\neq 4$ but $p+q\geq 3$, then the Lie algebra $\mathfrak{o}(n,\mb{C})$ is simple, thus $\textbf{g}$ is proportional to the Killing form. If $p+q=4$, then because $\mathfrak{o}(4)$ is simple, and $\Theta$ is a Cartan involution of $\mathfrak{o}(p,q)$ and of $\textbf{g}$, then it follows that $\textbf{g}$ is again proportional to the Killing form. 

Finally we note that the setup above is really just a tensor action by viewing $f^{\mb{C}}\in End(T_pM^{\mb{C}})$ as a tensor in the tensor product $v^{\mb{C}}\in T_pM^{\mb{C}}\otimes T_pM^{\mb{C}}$ w.r.t to a $O(n,\mb{C})$-module isomorphism $f^{\mb{C}}\mapsto v^{\mb{C}}$. Indeed let $f^{ij}\in End(T_pM^{\mb{C}})$ be defined by the matrix $(f^{ij})_{ij}=1$, and otherwise zero, w.r.t the basis: $Y:=\{y_1,\dots, y_n\}$ defined above. Then $\{f^{ij}\}_{ij}$ running over all $1\leq i,j\leq n$ form a basis for $End(T_pM^{\mb{C}})$. We define an isomorphism: $$End(T_pM^{\mb{C}})\xrightarrow{\phi^{\mb{C}}} T_pM^{\mb{C}}\otimes T_pM^{\mb{C}}, \ \ f^{ij}\mapsto y_i\otimes y_j.$$ An easy calculation shows that $\phi^{\mb{C}}(gfg^{-1})=g\cdot \phi^{\mb{C}}(f)$, where $g$ acts on tensors by $g\cdot (v_1\otimes v_2):= g(v_1)\otimes g(v_2)$, i.e $\phi^{\mb{C}}$ is an isomorphism of $O(n,\mb{C})$-modules. An easy calculation shows that $\phi^{\mb{C}}$ maps $End(T_pM)\mapsto T_pM\otimes T_pM$ by noting that $$\{f^{ij}|1\leq j\leq p\}\cup \{if^{ij}|p+1\leq j\leq n\}$$ is a basis for $End(T_pM)$. Trivially it maps $End(W)\mapsto W\otimes W$. It remains to show that it also maps $End(T_p\tilde{M})\mapsto T_p\tilde{M}\otimes T_p\tilde{M}$. To see this we note that since $T_p\tilde{M}$ is compatible with $W$, then we may choose a pseudo-orthogonal basis: $\{\tilde{e}_1,\dots, \tilde{e}_{\tilde{p}},\dots, \tilde{e}_n\}$ of $\tilde{g}$, and define analogously a map: $$\tilde{\phi}^{\mb{C}}: \tilde{f}^{ij}\mapsto \tilde{y}_i\otimes \tilde{y}_j,$$ w.r.t the real basis: $\{\tilde{y}_1,\dots, \tilde{y}_n\}:=\{\tilde{e}_1,\dots, \tilde{e}_{\tilde{p}}, i\tilde{e}_{\tilde{p}+1},\dots, i\tilde{e}_n\}$ of $W$. Thus let $g\in O(n)$ be the map sending $y_j\mapsto \tilde{y}_j$, then $\tilde{f}^{ij}=gf^{ij}g^{-1}$, i.e $${\phi}^{\mb{C}}(\tilde{f}^{ij})={\phi}^{\mb{C}}(gf^{ij}g^{-1})=g\cdot {\phi}^{\mb{C}}(f^{ij}):=g(y_i)\otimes g(y_j)=\tilde{y}_i\otimes \tilde{y}_j=\tilde{\phi}^{\mb{C}}(\tilde{f}^{ij}).$$ Thus since $\tilde{\phi}^{\mb{C}}$ maps analogously $End(T_p\tilde{M})$ into $T_p\tilde{M}\otimes T_p\tilde{M}$, then so does $\phi^{\mb{C}}$.
Therefore we conclude that the map $\phi^{\mb{C}}$ also induce an isomorphism of $O(p,q), \ O(\tilde{p},\tilde{q})$ and $O(n)$ modules respectively. 

We explore the induced isometry action of $O(n,\mb{C})$ on a more general tensor product space in Section 6. 
\\

The motivation behind the definition of a standard Wick-rotation comes from the following lemma together with results from real GIT.

\begin{lem} [\cite{W2}, Lemma 3.6] \label{op} The triple of real forms: $\Big{(} \mathfrak{o}(p,q), \mathfrak{o}(\tilde{p},\tilde{q}), \mathfrak{o}(n)\Big{)}$, embedded into $\mathfrak{o}(n,\C)$ under a standard Wick-rotation is a compatible triple of Lie algebras.     \end{lem}

We begin by observing that the definition of a Wick-rotation is in fact equivalent to the definition of a standard Wick-rotation, i.e we may always find such an embedding of the tangent spaces, we only need to use the following lemma:

\begin{lem} \label{wick} Let $\Big{(}\mathbb{C}^n, \langle-,-\rangle\Big{)}$ be the standard holomorphic inner product space, i.e $\langle Z_1,Z_2 \rangle:= \sum_{i=1}^n z^1_iz^2_i$ for any $Z_1:=(z_1^1,\dots, z_1^n)\in \mathbb{C}^n\ni Z_2:=(z_2^1,\dots, z_2^n)$. Then there exist a compatible triple: $\Big{(}\mathbb{R}^n(p,q), \mathbb{R}^n(\tilde{p},\tilde{q}), \mathbb{R}^n(n,0)\Big{)}$ of any signatures $p+q=\tilde{p}+\tilde{q}=n+0=n$.\end{lem}
\begin{proof} For a signature $p+q=n$, there is a conjugation map $\mathbb{C}^n\rightarrow \mathbb{C}^n$, defined by $Z\mapsto I_{p,q}\bar{Z}$ where $I_{p,q}$ is the diagonal matrix with diagonal entries: $(+1,\dots, +1, -1,\dots, -1)$ ($+1$ $p$-times, $-1$ $q$-times). It gives rise to a real slice $\mathbb{R}(p,q)\subset \mathbb{C}^N$, so because $[I_{p,q},I_{\tilde{p},\tilde{q}}]=0$ we have a compatible triple: $$\Big{(}\mathbb{R}^n(p,q), \mathbb{R}^n(\tilde{p},\tilde{q}), \mathbb{R}^n(n,0)\Big{)}.$$ The lemma is proved.\end{proof} 

\begin{cor} \label{sw} If $M$ and $\tilde{M}$ are Wick-rotated at $p\in M\cap \tilde{M}$, then they are also Wick-rotated by a standard Wick-rotation.  \end{cor} 
\begin{proof} Let $M$ and $\tilde{M}$ be Wick-rotated at $p\in M\cap \tilde{M}$. By Lemma \ref{wick} and since $(T_pM)^{\mathbb{C}}\cong \mathbb{C}^n$ as holomorphic inner product spaces, then we can also find a real slice $V$ of $(T_pM)^{\mathbb{C}}$ with signature $(\tilde{p},\tilde{q})$, such that $T_pM$ and $V$ form a compatible triple with a compact real slice $W$. Thus we can extend a real isomorphism: $V\xrightarrow{\psi} T_p\tilde{M}$, to an isomorphism $T_pM^{\mathbb{C}}\rightarrow T_pM^{\mathbb{C}}$, such that $\Big{(}\psi^{-1}(T_p\tilde{M}), T_p{M}, W\Big{)}$ form a compatible triple. This proves that $M$ and $\tilde{M}$ are Wick-rotated by a standard Wick-rotation. The corollary is proved.\end{proof}

Thus the results from \cite{W2}, \cite{W1} hold for Wick-rotated spaces, and we shall therefore always assume a Wick-rotation instead of a standard Wick-rotation. 

\subsection{Real GIT for semi-simple groups}

\textsl{Convention}: For a Lie group $G$ which has finitely many connected components we say $G$ is fcc.

\

Let $G$ be a real semi-simple linear group which is fcc, and $G\xrightarrow{\rho^G_V} GL(V),$ be a real representation. Denote: $G=Ke^{\mathfrak{p}}$, to be the Cartan decomposition w.r.t a global Cartan involution: $G\xrightarrow{\Theta}G$, where $\mathfrak{g}=\mathfrak{k}\oplus \mathfrak{p}$ is the Cartan decomposition of $\mathfrak{g}$ w.r.t $d\Theta:=\theta.$ Let $\langle-,-\rangle$ be a $K$-invariant inner product on $V$ such that $d\rho^G_V(\mathfrak{p})$ consists of symmetric operators w.r.t $\langle-,-\rangle$. A vector $v\in V$ is said to be a \emph{minimal vector} if: $$(\forall g\in G)(||g\cdot v||\geq ||v||),$$ the set of minimal vectors shall be denoted by $\mathcal{M}(G,V)\subset V$.

\

The following theorem by \textsl{Richardson} and \textsl{Slodowy} in \cite{RS}, which relates the closure of a real orbit to the existence of a minimal vector, is worth mentioning:

\begin{thm}[RS] \label{RS} The following statements hold:
\begin{enumerate}
\item{} A real orbit $Gv$ is closed if and only if $Gv\cap \mathcal{M}(G,V)\neq\emptyset$.
\item{} If $v$ is a minimal vector then $Gv\cap \mathcal{M}(G,V)=Kv$.
\item{} If $Gv$ is not closed then there exist $p\in \mathfrak{p}$ such that $e^{tp}\cdot v\rightarrow \alpha\in V$ exist as $t\rightarrow \infty$, and $G\alpha\subset V$ is closed. Moreover $G\alpha\subset \overline{Gv}$ is the unique closed orbit in the closure. 
\item{} A vector $v\in V$ is minimal if and only if $\Big{(} \forall x\in \mathfrak{p}\Big{)}\Big{(}\langle x\cdot v,v  \rangle=0\Big{)}$, where $x\cdot v$ is the differential action $d\rho^G_V(x)(v)$.
\end{enumerate}
  \end{thm}

Parts (1), (2) and (4) of the theorem is known as the \emph{Kempf-Ness Theorem}, for which it was first proved for linearly complex reductive groups. One shall also remark that Theorem \ref{RS} also holds for a more general class of real reductive Lie groups which includes the class of semi-simple linear groups which are fcc (\cite{CR}).

\

We also recall:

\begin{defn} Let $G^{\mb{C}}$ be a complex Lie group. A closed real Lie subgroup $G\subset (G^{\mb{C}})_{\mb{R}}$ is said to be a \emph{real form}, if $\mathfrak{g}\subset \mathfrak{g}^{\mb{C}}$ is a real form, and $G^{\mb{C}}=G\cdot G^{\mb{C}}_0$ (abstract group product). If $U\subset G^{\mb{C}}$ is a real form which is compact, then we shall say it is a \emph{compact real form}. \end{defn}

Note that $G$ is fcc if and only if $G^{\mb{C}}$ is fcc, and moreover if $G^{\mb{C}}$ is fcc, and $U$ a compact real form, then $U$ must be a maximally compact subgroup of $G^{\mb{C}}$.

\

For a real form $G\subset G^{\mb{C}}$, a complex action $G^{\mb{C}}\xrightarrow{\rho^{\mb{C}}} GL(V^{\mb{C}})$ is a \textsl{complexified action} of a real action: $G\xrightarrow{\rho^G_V} GL(V)$, if $\rho^{\mb{C}}(G)(V)=\rho^G_V(G)(V)$. Let $G$ be semi-simple and the notation as above, then if $\tau$ denotes the conjugation map of the compact real form: $\mathfrak{u}:=\mathfrak{k}\oplus i\mathfrak{p}\subset \mathfrak{g}^{\mb{C}}$, then $\tau$ restricted to $\mathfrak{g}$ is precisely $\theta$. If $(G^{\mb{C}})_{\mb{R}}=Ue^{i\mathfrak{u}}$ is the corresponding Cartan decomposition w.r.t $\tau$, then it is possible to choose a $U$-invariant Hermitian inner product: $H(-,-)$ on $V^{\mb{C}}$ which is compatible with $V$, note that $K\subset U$, and we have that: $$\mathcal{M}(G,V)\subset \mathcal{M}(G^{\mb{C}},V^{\mb{C}}).$$ 

\

Let $G\subset GL(V)$ ($V$ a real vector space) be a semi-simple linear Lie group which is fcc, and $G^{\mb{C}}\subset GL(V^{\mb{C}})$ be the Zariski-closure of $G$. We recall the following known result:

\begin{thm} [\cite{EJ}, Lemma 2.2 + Remark p.3] \label{clos} If $v\in V$, then $Gv\subset V$ is closed if and only if $G^{\mb{C}}v\subset V^{\mb{C}}$ is closed. Also $G^{\mb{C}}v\cap V$ is a finite disjoint union of real orbits: $Gv_j\subset V$.  \end{thm}

If $U\subset G^{\mb{C}}\supset G$ are compatible real forms with $U$ a compact real form, together with real representations: $\rho^G_V$ and $\rho^U_W$ which have the same complexification, and $V, W$ are compatible real forms of $V^{\mb{C}}$ then the following result hold:

\begin{thm} [\cite{W2}] \label{W1r} Assume the assumptions above. Then there exist $v\in V$ and $w\in W$ such that $Uw\subset G^{\mb{C}}v\supset Gv$ if and only if $Gw\cap Gv\neq \emptyset$.  \end{thm}

We end the section with an example. Consider the notation of the example in the previous section (paragraph after Defn \ref{sW}), i.e the conjugation action: $$O(n,\mb{C})\rightarrow GL\Big{(}End(T_pM^{\mb{C}})\Big{)}, \ \ g\cdot f:=gfg^{-1}, \ \ n\geq 3.$$ Put the $O(n)$-invariant Hermitian inner product: $H:=\textbf{g}^{\mb{C}}(\cdot,\mathcal{T}(\cdot))$ on $End(T_pM^{\mb{C}})$, where $\mathcal{T}$ is the conjugation map: $\mathcal{T}:=f\mapsto \tau f\tau$ of $End(W)\subset End(T_pM^{\mb{C}})$, and $\tau$ is the conjugation map of the compact real slice $W\subset T_pM^{\mb{C}}$. It is not difficult to see that $f\in End(T_pM^{\mb{C}})$ is a minimal vector if and only if $$\textbf{g}^{\mb{C}}(x, [f_+,f_-])=0, \forall x\in i\mathfrak{o}(n),$$ where $f=f_++f_-$ is the eigenspace decomposition w.r.t to $\mathcal{T}$. Thus the closed orbits: $O(n,\mathbb{C})\cdot f$, are precisely those which intersect $\mathcal{M}(O(n,\mb{C}), End(T_pM^{\mb{C}}))$. If we moreover restrict our vector space to the Lie algebra: $\mathfrak{o}(n,\mb{C})$, then the action is just the adjoint action, and thus the minimal vectors are precisely those $f\in \mathfrak{o}(n,\mb{C})$ satisfying $[f_+,f_-]=0$. 

\subsection{Real GIT for linearly real reductive groups}

In this subsection we shall extend Theorem \ref{clos} to real forms: $G\subset G^{\mb{C}}$, which are linearly real reductive. 

\

\begin{rem}Note that in the definition of a real form, although $G^{\mb{C}}\subset GL(V^{\mb{C}})$ for some real vector space $V$, then $G$ is not necessarily contained in $G^{\mb{C}}\cap GL(V)$. For example: $SU(2)\subset SL_2(\mb{C})\subset GL_2(\mathbb{C})$, but $SU(2)$ is not contained in $SL_2(\mb{C})\cap GL_2(\mathbb{R})=SL_2(\mb{R})$. However since $SL_2(\mb{C})$ is the universal complexification group of $SU(2)$, then we may find a real vector space $V$ such that $SU(2)\subset GL(V)$, and $SL_2(\mb{C})\subset GL(V^{\mb{C}})$, but this is not part of our assumptions in the definition of a real form. \end{rem}

\begin{defn} [\cite{Neeb}]  A \textsl{linearly complex reductive Lie group} $G^{\mb{C}}$ is a complex Lie group containing a compact subgroup $U$ such that $G^{\mb{C}}$ is the universal complexification group of $U$.  \end{defn}

In fact the complex Lie groups $G^{\mb{C}}$ which are fcc and have a compact real form are precisely the linearly complex reductive groups. Thus this class of groups are all self-adjoint by (\cite{Mos55}, Lemma 5.1), and so the class of groups lends itself to Theorem \ref{RS} by (\cite{CR}). All such groups $G^{\mb{C}}$ are algebraic (canonically), and so are fcc (\cite{RS}, 8.3). One should also note that a complex Lie group $G^{\mb{C}}$ which is fcc and has a reductive Lie algebra is linearly complex reductive if and only if $Z(G^{\mb{C}}_0)_0\cong (\mathbb{C}^{\times})^k$ (a complexified tori), see for example (\cite{Neeb}, Chapter 15). 

\begin{defn} A real linear group $G$ shall be called \textsl{linearly real reductive} if $G$ is fcc and $G_0$ is linearly real reductive in the sense of (\cite{Neeb}, Definition 16.1.4), i.e $Z(G_0)$ is compact and $\mathfrak{g}$ is reductive. \end{defn}

Thus $G$ is also a real reductive Lie group in the sense of (\cite{CR}), i.e there is a faithful representation with closed image: $G\subset GL(V)$, together with a global Cartan involution of $GL(V)$ leaving $G$ invariant.

All semi-simple complex Lie groups are linearly complex reductive, and all real semi-simple linear groups which are fcc are linearly real reductive. One shall also note that the class of linearly real reductive Lie groups $G$ are precisely the Lie groups (fcc) which are completely reducible (i.e every representation is completely reducible). 

In contrary to semi-simple real forms, not all real forms of a linearly complex reductive group are linearly real reductive. Indeed take $G^{\mb{C}}:= \mathbb{C}^{\times}$, then it is linearly complex reductive, with a compact real form $U\cong S^1$. But $G:=\mathbb{R}^{\times}$ is also a real form, however it is not linearly real reductive, but it's nevertheless a real reductive Lie group in the sense of (\cite{CR}). We also see that if $G$ is linearly real reductive and is a real form of some complex group $G^{\mb{C}}$, then $Z(G^{\mb{C}}_0)_0$ has compact real form: $Z(G_0)_0$, which must be a torus, and thus $G^{\mb{C}}$ is linearly complex reductive.  

\

The following extends a property of semi-simple groups:

\begin{lem} \label{seq} Suppose $G$ is linearly real reductive. Then the image of $G$ under any real representation $\rho^G_V(G)\subset GL(V)$ is closed. \end{lem}
\begin{proof} Since $G$ is fcc then we may assume w.l.o.g that $G$ is connected. Now the Lie algebra $\mathfrak{g}$ is reductive thus $\mathfrak{g}=\mathfrak{g}'\oplus \mathfrak{z}(\mathfrak{g})$. Now $Z(G)\subset G$ is compact and has Lie algebra: $\mathfrak{z}(\mathfrak{g})$. Let $G'\subset G$ be the unique connected Lie subgroup of $G$ with Lie algebra $\mathfrak{g}'$. Then $G'$ is semi-simple and connected, and since $G$ is connected then it is generated by $\langle G', Z(G)\rangle$. Also since $Z(G)$ is compact then the image $H'':=\rho^G_V(Z(G))$ is compact, and by (\cite{Neeb}, Corollary 14.5.7), the image $H':=\rho^G_V(G')\subset GL(V)$ is closed. The image $H:=\rho^G_V(G)$ is generated by $H'$ and $H''$. Recall that the topology of $GL(V)\subset End(V)$ is a metric subspace with an induced norm metric: $d(-,-)$ on $End(V)$, satisfying $d(gh,0)\leq d(g,0)d(h,0)$ for all $g,h\in End(V)$. Now suppose $(y_n)\subset H$ is any convergent sequence in $GL(V)$. Then clearly $y_n=a_nb_n$ for sequences $(a_n)\subset H'$ and $(b_n)\subset H''$. Thus since $H''$ is compact then $(b_n)$ is a bounded sequence, and so we may choose a subsequence $(b_{m(k)})$ converging to $\beta\in H''$. It follows that $(a_{m(k)})$ must converge as well using the norm metric, thus it converges for some $\alpha\in H'$. But then $\lim_{n\rightarrow \infty} (y_n)=\lim_{k\rightarrow \infty} (y_{m(k)})=\lim_{k\rightarrow \infty} (a_{m(k)}b_{m(k)})=\alpha\beta\in H.$ This shows that $\rho^G_V(G)\subset GL(V)$ is closed as required. The lemma is proved.  \end{proof}

We now extend Theorem \ref{clos} to the case where our real form is linearly real reductive:

\begin{prop} \label{linearly} Let $G\subset G^{\mb{C}}$ be a real form which is of type real linearly real reductive. Assume $G^{\mb{C}}\xrightarrow{\rho^{\mb{C}}} GL(V^{\mb{C}})$ is a complexified Lie group action of a real Lie group action: $G\xrightarrow{\rho^G_V} GL(V)$. Then Theorem \ref{clos} holds. \end{prop}
\begin{proof} Now since $G^{\mb{C}}$ is algebraic, and $\rho^{\mb{C}}$ is a rational representation w.r.t the algebraic structure (\cite{Dong}, Theorem 5.11), then the image $H^{\mb{C}}:=\rho^{\mb{C}}(G^{\mb{C}})$ is a complex algebraic subgroup of $GL(V^{\mb{C}})$. The group $H^{\mb{C}}$ is fcc since $G^{\mb{C}}$ is fcc, and is a linearly complex reductive group, since if $U\subset G^{\mb{C}}$ is a compact real form, then $\rho^{\mb{C}}(U)$ is a compact real form of $H^{\mb{C}}$. Now since $H$ is assumed to be ffc then $H:=\rho^G_V(G)\subset GL(V)$ is a real closed subgroup of $H^{\mb{C}}$ by Lemma \ref{seq}. Now if $Q$ is the Zariski-closure of $H$ in $H^{\mb{C}}$, then $H^{\mb{C}}_0\subset Q^0$ where $Q^0$ is the Zariski-connected component of $Q$. Also since $H^{\mb{C}}=H\cdot H^{\mb{C}}_0\subset H\cdot Q^0\subset Q$, then we have $H^{\mb{C}}=Q$, so $H$ is Zariski-dense in $H^{\mb{C}}$, and in particular $H^{\mb{C}}$ is defined over $\mathbb{R}$. Thus denote the real algebraic subgroup: $H^{\mb{C}}(\mathbb{R}):=H^{\mb{C}}\cap GL(V)\subset H^{\mb{C}}$ then it is a real form under the anti-holomorphic involution: $X\mapsto \overline{X}$. Also $H^{\mb{C}}(\mathbb{R})_0\subset H\subset H^{\mb{C}}(\mathbb{R})\subset H^{\mb{C}}$, because $H^{\mb{C}}(\mathbb{R})$ and $H$ have the same Lie algebras, and moreover note that $H\subset H^{\mb{C}}(\mathbb{R})$ is closed. Thus if we consider the identity representation: $H^{\mb{C}}\rightarrow GL(V^{\mb{C}})$, then we have exactly the assumptions in \cite{RS}, and we can mimic the proof of (\cite{EJ}, Lemma 2.2). But given $v\in V$ then $H^{\mb{C}}v:=G^{\mb{C}}v$ and $Hv:=Gv$ so the proposition follows.\end{proof}

We also make a note of the following theorem, which is well-known for semi-simple linear Lie groups which are fcc, and also holds for reductive algebraic groups in the context of rational representations (\cite{BC}). The theorem also applies to the class of linearly real reductive groups:

\begin{thm} \label{mos} Let $G\subset GL(E)$ be a linearly real reductive Lie group. Then the following statements hold:
\begin{enumerate}
\item{} There exist a global Cartan involution of $GL(E)$ leaving $G$ invariant. 
\item{} If $\mathfrak{gl}(E)\xrightarrow {\theta}\mathfrak{gl}(E)$ is a Cartan involution leaving $\mathfrak{g}$ invariant, then $\Theta(G)\subset G$, where $\Theta$ is the global Cartan involution of $GL(E)$ with differential $\theta$. 
\item{} All Cartan involutions of $G$ are conjugate by an inner automorphism of $G$.
\item{} Let $G\xrightarrow{\rho^G_V} GL(V)$ be a real representation. Then given any global Cartan involution $\Theta$ of $G$, then there exist a global Cartan involution $\Theta'$ of $GL(V)$ such that: $\rho^G_V(\Theta(g))=\Theta'(\rho^G_V(g)), \forall g\in G$.
\end{enumerate}
  \end{thm}
\begin{proof} Since the center: $\mathfrak{z}(\mathfrak{g})$, of $\mathfrak{g}$ is algebraic because $Z(G_0)$ is compact, then $\mathfrak{g}$ is also algebraic since it is a reductive Lie algebra, thus we can mimic the proof of (\cite{EJ}, Remark p.3). Therefore by the results of (\cite{RS}) cases (1), (2) and (3) follows. Case (4). Since $Z(G_0)$ is compact, then $\rho^G_V(Z(G_0))\subset GL(V)$ is an algebraic subgroup with Lie algebra $d\rho^G_V(\mathfrak{z}(\mathfrak{g}))$, and so the image $d\rho^G_V(\mathfrak{g})$ is an algebraic reductive subalgebra in $\mathfrak{gl}(V)$. Thus following the steps in the proof of (\cite{BC}, Proposition 13.5), case (4) follows. The theorem is proved.   \end{proof}

\begin{cor} \label{aab} Let $G\subset G^{\mb{C}}\supset U$ be two compatible real forms where $G$ is linearly real reductive, and $U$ is a compact real form. Suppose $G^{\mb{C}}\subset GL(V^{\mb{C}})$, then there exist a $U$-invariant Hermitian form on $V^{\mb{C}}$ such that $G^{\mb{C}}$ and $G$ are both self-adjoint.  \end{cor}
\begin{proof} Let $\mathfrak{u}\subset \mathfrak{g}^{\mb{C}}$ be the Lie algebra of $U$, i.e it is a compact real form of $\mathfrak{g}^{\mb{C}}$. By (\cite{Mos55}) the group $G^{\mb{C}}$ is self-adjoint w.r.t a Hermitian inner product $H(-,-)$ on $V^{\mb{C}}$. Let $\mathfrak{gl}(V^{\mb{C}})\xrightarrow{\tau} \mathfrak{gl}(V^{\mb{C}})$ be the conjugation map of $\mathfrak{gl}(V^{\mb{C}})$ with fix points: $\mathfrak{u}(n)$, leaving $\mathfrak{g}^{\mb{C}}$ invariant w.r.t $H(-,-)$. Now since $\mathfrak{g}$ is compatible with $\mathfrak{u}$, then $\mathfrak{g}=\mathfrak{k}\oplus \mathfrak{p}$, with $\mathfrak{k}\subset \mathfrak{u}$ and $\mathfrak{p}\subset i\mathfrak{u}$. Thus $\tau$ also leaves invariant $\mathfrak{g}$. Now by identifying the real groups: $\Big{(}GL(V^{\mb{C}})\Big{)}_{\mathbb{R}}\cong GL\Big{(}(V^{\mb{C}})_{\mathbb{R}}\Big{)}$, then $\tau$ induces a Cartan involution of $\mathfrak{gl}\Big{(}(V^{\mb{C}})_{\mathbb{R}}\Big{)}$ w.r.t the real part of $H(-,-)$, leaving the copy of $\mathfrak{g}\hookrightarrow \mathfrak{gl}\Big{(}(V^{\mb{C}})_{\mathbb{R}}\Big{)}$ invariant. Thus the corresponding global Cartan involution of $GL\Big{(}(V^{\mb{C}})_{\mathbb{R}}\Big{)}$ leaves the copy $G\hookrightarrow GL\Big{(}(V^{\mb{C}})_{\mathbb{R}}\Big{)}$ invariant by (2) of Theorem \ref{mos}, and so the global conjugation map of $GL(V^{\mb{C}})$ with differential $\tau$ must also leave the original copy of $G$ invariant. The corollary is proved. \end{proof}

\begin{rem} Let $V\subset (V^{\mb{C}}, g^{\mb{C}})$ be a real form of a holomorphic inner product space, and consider the linear isometry groups: $G:=O(p,q)\subset G^{\mb{C}}:=O(n,\mb{C})$. Then as a Lie group $G^{\mb{C}}$ is linearly complex reductive for all $n\geq 1$, and for $n>2$ the real form $O(p,q)$ is semi-simple, while for $n=1$ the group $G$ is finite thus is linearly real reductive. For $n=2$ then $G$ is not linearly real reductive, but is the real points of $O(2,\mb{C})$, i.e is a reductive algebraic group, thus the group satisfies the assumptions of the setup in \cite{RS}. Therefore all the results obtained here in this section, can also be applied to a real form: $O(p,q)\subset O(n,\mb{C})$ for all $p+q=n$.  \end{rem}

In regards to Wick-rotations we are mainly interested in the real forms: $O(p,q)\subset O(n,\mb{C})\supset O(\tilde{p},\tilde{q})$, for $p+q=\tilde{p}+\tilde{q}=n$.

\section{Balanced representations}

Throughout sections 3, 4 and 5, when considering a complex Lie group $G^{\mb{C}}$ it shall always be of type linearly complex reductive. Moreover a real form $G\subset G^{\mb{C}}$ shall always be assumed to be either linearly real reductive or in the case where $G^{\mb{C}}$ is defined over $\mb{R}$, the real points: $G=G_{\mb{R}}$. The groups to have in mind are $O(p,q)\subset O(n,\mb{C})$.  
 
\begin{defn} [\cite{RS}, \textsl{Section} 5.2] Let $G\xrightarrow{\rho^G_V} GL(V)$ be a real representation, then $\rho^G_V$ is said to be \emph{balanced representation} if there exist an involution $V\xrightarrow{\theta} V$, and a global Cartan involution: $G\xrightarrow{\Theta} G$ such that: $$\Big{(}\forall g\in G \Big{)}\Big{(}\rho^G_V(\Theta(g))=\theta\circ\rho^G_V(g)\circ\theta\Big{)}.$$ \end{defn}

For example in the case of the adjoint action of a semi-simple Lie group $G$, then the involutions balancing the action are precisely: $\pm \theta$, where $\mathfrak{g}\xrightarrow{\theta}\mathfrak{g}$ is a Cartan involution of the Lie algebra: $\mathfrak{g}$. Note also that any real representation: $U\rightarrow GL(W)$ of a compact Lie group $U$ is balanced, since the global Cartan involution of $U$ is $1_U$ and thus $1_W$ is an involution balancing the action. 
\\

It is also worth noting that if our group $G$ has the property that a global Cartan involution of $G$: $\Theta=Ad(k)$ for some $k\in K$ of order $2$, then all representations are naturally balanced, since one may take $V\xrightarrow{\rho^G_V(k)} V$ as a natural choice of involution balancing a representation $\rho^G_V$. This is the case for example with the pseudo-orthogonal groups: $O(p,q)$. The group $SL_2(\mathbb{R})$ does not have this property for instance.

It is not difficult to see that an involution $\theta$ balancing a representation gives rise to a $G$-invariant symmetric non-degenerate bilinear form: $\langle-,-\rangle$ on $V$ such that $\theta$ is a Cartan involution, i.e $\langle v,\theta(v)\rangle>0$ for all $v\neq0$, (see for example \cite{RS}, \textsl{Section} 5.2). Note that $-\theta$ is also an involution balancing the action, and $\theta$ can not be conjugate to $-\theta$ by the action of $G$.

\begin{defn} \label{d} Let $G\xrightarrow{\rho^G_V} GL(V)$ be a balanced real representation and $\theta$ an involution balancing $\rho^G_V$. Let $\langle-,-\rangle$ be a $(G,\theta)$-invariant symmetric non-degenerate bilinear form on $V$. Then any Cartan involution $\theta'$ of $\langle-,-\rangle$ is said to be an \textsl{inner Cartan involution} of $\rho^G_V$ if it is conjugate by the action of $G$ to $\theta$.\end{defn}

In the case of the adjoint action for semi-simple groups, then fixing the Killing form: $-\kappa(-,-)$ on $\mathfrak{g}$, the inner Cartan involutions are precisely the Cartan involutions contained in $Aut(\mathfrak{g})$.

\

\begin{rem} Hereon whenever we consider a balanced representation $\rho^G_V$ we shall always fix a $G$-invariant non-degenerate symmetric bilinear form: $\langle-,-\rangle$ on $V$, and speak of the inner Cartan involutions of $\langle-,-\rangle$.\end{rem}

\

We shall also consider complex representations, and therefore define analogously:

\begin{defn} Suppose $G^{\mathbb{C}}\xrightarrow{\rho^{\mathbb{C}}}GL(V^{\mathbb{C}})$ is a complex representation. We say $\rho^{\mb{C}}$ is balanced if $(\rho^{\mb{C}})_{\mb{R}}$ (the real representation) is balanced w.r.t a conjugation map: $V^{\mb{C}}\xrightarrow{\tau} V^{\mb{C}}$. \end{defn}

Note that the definition is a generalisation of the adjoint action of semi-simple Lie groups to general actions. We also extend Definition \ref{d} to balanced complex actions $\rho^{\mb{C}}$, i.e if $\tau$ balances $\rho^{\mb{C}}$ then any real involution: $\rho^{\mb{C}}(g)\tau\rho^{\mb{C}}(g^{-1})$ for some $g\in G^{\mb{C}}$ shall be called an inner Cartan involution of $V^{\mb{C}}$. One observes that given a $\tau$ which balances a complex action, then we may choose a $G^{\mb{C}}$-invariant Hermitian form $H(-,-)$ on $V^{\mb{C}}$. 

\

In the case where $V^{\mb{C}}$ is an irreducible $\mathfrak{g}^{\mb{C}}$-module there are restrictions on the involutions balancing the representation:

\begin{prop} \label{bo} Let $G^{\mb{C}}\xrightarrow{\rho^{\mb{C}}} GL(V^{\mb{C}})$ be a balanced complex representation. Assume that $V^{\mb{C}}$ is an irreducible $\mathfrak{g^{\mb{C}}}$-module. Then any two real involutions: $V^{\mb{C}}\xrightarrow{{\tau}, \tilde{\tau}} V^{\mb{C}}$ balancing $\rho^{\mb{C}}$ are conjugate by the action of $G^{\mb{C}}_0$ up to scaling of $\pm1$.  \end{prop}
\begin{proof} Assume $\tau$ and $\tilde{\tau}$ are two conjugation maps which balances $\rho^{\mb{C}}$, so there exist global Cartan involutions: $\Theta, \tilde{\Theta}$ of $G^{\mb{C}}$, such that: $$\rho(\Theta(g))=\tau\circ \rho(g)\circ \tau, \ \ \ \rho(\tilde{\Theta}(g))=\tilde{\tau}\circ \rho(g)\circ \tilde{\tau}, \forall g\in G^{\mb{C}}.$$ Now since $\Theta$ and $\tilde{\Theta}$ are conjugate by an inner automorphism of $G^{\mb{C}}$, then it is not difficult to see that there exist $h\in G^{\mb{C}}_0$, such that $$\rho(h)\tau\rho(h^{-1})\circ\tilde{\tau}\circ\rho(g)=\rho(g)\circ\rho(h)\tau\rho(h^{-1})\circ \tilde{\tau}, \forall g\in G^{\mb{C}}.$$ Thus $f:=\rho(h)\tau\rho(h^{-1})\circ\tilde{\tau}$ is a complex linear map which is a $G^{\mb{C}}$-module isomorphism, using the exponential map this is also a $\mathfrak{g^{\mb{C}}}$-module isomorphism on Lie algebra level, i.e for the differential action: $\mathfrak{g^{\mb{C}}}\xrightarrow{d\rho^{\mb{C}}} \mathfrak{gl}(V^{\mb{C}})$. Now since $V^{\mb{C}}$ is irreducible, then by Schur's lemma we must have that $f=\lambda 1_{V^{\mb{C}}}$, for some $\lambda\in\mathbb{C}$. Thus $\lambda^2=1$ since $\lambda^21_{V^{\mb{C}}}=(\lambda\tilde{\tau})^2=(\rho(h)\tau(\rho(h))^{-1})^2=1_{V^{\mb{C}}}$, and so the proposition is proved. \end{proof}

 \begin{rem} \label{q2} Note that Proposition \ref{bo} fails in the case of the trivial representation, indeed any conjugation map $\sigma$ of $V^{\mb{C}}\neq 0$ will balance the trivial representation. So if $\sigma$ is a conjugation map of $V^{\mb{C}}$ then $\sigma$ and $i\sigma$ are not conjugate by the action of $G^{\mb{C}}$ up to $\pm 1$. In general it even fails for a non-trivial reducible representation as well. Indeed let $Ad$ be the adjoint action, then it is non-trivial, and $\tau$ be a conjugation map of a compact real form, then $\tau$ will balance $Ad$. Consider the representation $0_{V^{\mb{C}}}\oplus Ad$ for $V^{\mb{C}}$ any non-zero complex vector space. Then this is a non-trivial representation, and for example if $\sigma$ is any conjugation map of $V^{\mb{C}}$ then the two involutions: $\sigma\oplus \tau$ and $i\sigma\oplus \tau$ both balance this representation, however they cannot be conjugated by the action of $G^{\mb{C}}$ up to $\pm 1$. \end{rem}
  
A complex action is balanced in the following sense:

\begin{prop} \label{c} Suppose $G^{\mathbb{C}}\xrightarrow{\rho^{\mathbb{C}}}GL(V^{\mathbb{C}})$ is a complex representation. Then there is a compact real form: $U\subset G^{\mathbb{C}}$, and a real form $W\subset V^{\mathbb{C}}$, such that $\rho^{\mathbb{C}}(U)(W)\subset W$ if and only if $\rho^{\mb{C}}$ is balanced. \end{prop}
\begin{proof} Suppose a compact real form: $U\subset G^{\mathbb{C}}$ restricts to an action on a real form: $W\subset V^{\mathbb{C}}$. Denote $\Theta$ for the corresponding Cartan involution of $(G^{\mathbb{C}})_{\mathbb{R}}$ with fix points $U$. Then clearly: $$\tau(\rho^{\mathbb{C}}(u)(w_1+iw_2))=\rho^{\mathbb{C}}(u)(w_1-iw_2)=\rho^{\mathbb{C}}(u)(\tau(w_1+iw_2)).$$ Also if $g:=e^{ix}$ for $x\in \mathfrak{u}$ (the Lie algebra of $U$), then: $$\rho^{\mathbb{C}}(e^{ix})(w)=\sum_{n:=2k} \frac{1}{n!} (d\rho^{\mathbb{C}})^n(ix)(w) +\sum_{n:=2k+1} \frac{1}{n!} (d\rho^{\mathbb{C}})^n(ix)(w)=w'_1+iw'_2,$$ for $w'_1,w'_2\in W$. Thus, $$\rho^{\mathbb{C}}(e^{-ix})(w)=\sum_{n:=2k} \frac{1}{n!} (d\rho^{\mathbb{C}})^n(ix)(w) -\sum_{n:=2k+1} \frac{1}{n!} (d\rho^{\mathbb{C}})^n(ix)(w)=w'_1-iw'_2,$$ and so $\tau(\rho^{\mathbb{C}}(e^{-ix})(w))=\rho^{\mathbb{C}}(e^{ix})(w)$. So since $(V^{\mathbb{C}})_{\mathbb{R}}=W\oplus iW$, then it follows that $\rho^{\mathbb{C}}$ is balanced w.r.t $\tau$. Conversely this is clear, since if the action is balanced then one has the equation: $$\rho^{\mathbb{C}}(\Theta(g))=\tau\circ\rho^{\mathbb{C}}(g)\circ \tau,$$ where $\Theta$ is some Cartan involution of $(G^{\mathbb{C}})_{\mathbb{R}}$, and $\tau$ is some conjugation map in $V^{\mathbb{C}}$. Denote $U$ for the compact real form of $G^{\mathbb{C}}$ which is the fix points of $\Theta$, and $W$ for the real form of $V^{\mathbb{C}}$, which is the fix points of $\tau$, then clearly $\rho^{\mathbb{C}}(U)(W)\subset W$ as required. The proposition is proved.  \end{proof}

In other words a complex action is balanced if and only if it is a complexified action of a real action of a compact real form. An example of a complex action which is not a complexification of any real action of a compact real form, is the faithful action of $G^{\mb{C}}:=SL_2(\mb{C})$ on $V^{\mb{C}}:=\mb{C}^2$ by $X\cdot v:=Xv$. Indeed it is enough to show it for the compact real form: $U:=SU(2)\subset G^{\mb{C}}$ (as all compact real forms are isomorphic). If this was the case, then the restricted action of $SU(2)$ on a real form $W\subset V^{\mb{C}}$ would also be faithful locally, and thus we could embed $\mathfrak{su}(2)\hookrightarrow \mathfrak{gl}(2,\mathbb{R})$. However all semi-simple Lie subalgebras of $\mathfrak{gl}(2,\mathbb{R})$ are contained in $\mathfrak{sl}_2(\mathbb{R})$, and hence we would obtain: $\mathfrak{su}(2)\cong \mathfrak{sl}_2(\mathbb{R})$, which is false. 
It is however a complexified action of the real form: $G:=SL_2(\mathbb{R})\subset G^{\mb{C}}$ acting on $V:=\mathbb{R}^2\subset V^{\mb{C}}$, which is also non-balanced, indeed if it were balanced then $\mathfrak{sl}_2(\mathbb{R})\cong d\rho^G_V(\mathfrak{g})\subset \mathfrak{o}(p,q)$ is a Lie subalgebra for some $p+q=2$, this is impossible, as $\mathfrak{o}(2)$ and $\mathfrak{o}(1,1)$ are both abelian.

Note that this example can be generalised to the faithful action of $SL_n(\mb{C})$ acting on $V^{\mb{C}}:=\mathbb{C}^n$ for any $n\geq 2$.

Recall that two representations: $G_1\xrightarrow{\rho^{G_1}_{V_1}} GL(V_1)$, and $G_2\xrightarrow{\rho^{G_2}_{V_2}} GL(V_2)$ are said to be \textsl{isomorphic} if there are Lie group isomorphisms: $G_1\xrightarrow{\psi_1} G_2$ and $GL(V_1)\xrightarrow{\psi_2} GL(V_2)$, such that: $\rho^{G_1}_{V_1}=\psi_2\circ \rho^{G_2}_{V_2}\circ \psi_1$. We write $\rho^{G_1}_{V_1}\cong \rho^{G_2}_{V_2}$. 

\begin{cor} \label{q1} Let $G^{\mathbb{C}}\xrightarrow{\Psi} GL(V^{\mathbb{C}})$ be a complex representation. Assume $V^{\mb{C}}$ is an irreducible $\mathfrak{g}^{\mb{C}}$-module. Let $U\subset G^{\mathbb{C}}\supset \tilde{U}$ be compact real forms, and $W\subset V^{\mathbb{C}}\supset \tilde{W}$ be real forms. Suppose $U\xrightarrow{\rho^U_W} GL(W)$ and $\tilde{U}\xrightarrow{\rho^{\tilde{U}}_{\tilde{W}}} GL(\tilde{W})$ are two real representations with $\Psi=(\rho^U_W)^{\mathbb{C}}=(\rho^{\tilde{U}}_{\tilde{W}})^{\mathbb{C}}$. Then $\rho^U_W\cong \rho^{\tilde{U}}_{\tilde{W}}$. \end{cor}
\begin{proof} Two real representations: $\rho^U_W$ and $\rho^{\tilde{U}}_{\tilde{W}}$, with complexification $\Psi$, give rise to two balanced Cartan involutions: $\tau_W$ and $\tau_{\tilde{W}}$, namely the conjugation maps with fix points $W$ and $\tilde{W}$ respectively, by Proposition \ref{c}. Now following the proof of Proposition \ref{bo}, we know that there exist $g\in (G^{\mathbb{C}})_{\mathbb{R}}$ such that $\Psi(g)\circ\tau_W \circ \Psi(g^{-1})=\lambda\tau_{\tilde{W}}$ ($\lambda=\pm1$), with $Ad(g)(U):=gUg^{-1}=\tilde{U}$. We note that if $\lambda=-1$, then $\Psi(g^{-1})(i\tilde{W})=W$, and if $\lambda=1$, then $\Psi(g^{-1})(\tilde{W})=W$. However since the action $U\xrightarrow{\rho^U_{iW}} GL(iW)$ given by: $u\cdot iw:=i\rho^U_W(w)$ is isomorphic to $\rho^U_W$, then we can assume w.l.o.g that $\lambda=1$. Thus we have isomorphisms: $U\xrightarrow{Ad(g)} \tilde{U}$, and $GL(W)\xrightarrow{Ad(\Psi(g))} GL(\tilde{W})$, where $Ad(\Psi(g))(f):=\Psi(g)f\Psi(g^{-1})$. One easily checks that: $$\rho^U_W=Ad(\Psi(g))\circ \rho^{\tilde{U}}_{\tilde{W}}\circ Ad(g),$$ and thus proves the corollary. 

 \end{proof}

Let $O(p,q)\subset GL(V)$, be defined as the isometry group of some non-degenerate symmetric bilinear form: $\langle-,-\rangle$, of signature $p+q=Dim(V)$. Then for $\rho^G_V$ to be balanced is just a stronger version of Theorem \ref{mos} (case 4), i.e we may choose $\Theta'$ to be a Cartan involution of an $O(p,q)$ group:  

\begin{prop} Let $G\xrightarrow{\rho^G_V}GL(V)$ be a real representation. Then $\rho^G_V$ is balanced if and only if there exist a pseudo-orthogonal group $O(p,q)\subset GL(V)$ and a Cartan involution of $O(p,q)$ leaving $\rho^G_V(G)$ invariant.  \end{prop} 
\begin{proof} Suppose $V\xrightarrow{\theta} V$ is an involution balancing $\rho^G_V$ w.r.t $\Theta$ of $G$, let $\langle-,-\rangle$ be a $\Big{(}G,\theta\Big{)}$-invariant symmetric non-degenerate bilinear form of some signature $p+q=n$. Denote $O(p,q)\subset GL(V)$, for the isometry group of $\langle-,-\rangle$, then $\rho^G_V(G)\subset O(p,q)\subset GL(V)$. Now $g\mapsto \theta\circ g\circ \theta$, is a global Cartan involution of $O(p,q)$, thus $$\rho^G_V(g)\mapsto \theta\circ\rho^G_V(g)\circ \theta=\rho^G_V(\Theta(g))\in \rho^G_V(G),$$ for the fixed global Cartan involution $\Theta$ of $G$. Conversely suppose there exist a pseudo-orthogonal group $O(p,q)\subset GL(V)$ and a global Cartan involution $\Theta'$ of $O(p,q)$ leaving $\rho^G_V(G)$ invariant. Note that $p,q\neq 1$. Let $\langle-,-\rangle$ be the symmetric non-degenerate bilinear form of signature $p+q=n$ associated to $O(p,q)$. Then let $\theta$ be any Cartan involution of $\langle-,-\rangle$ w.r.t $\Theta'$, i.e it balances the isometry action of $O(p,q)$ on $V$. Now let $\Theta$ be a global Cartan involution of $G$, and let $\Theta_1$ be a global Cartan involution of $GL(V)$ extending $\Theta'$, by Theorem \ref{mos}. Also there exist a global Cartan involution $\Theta_2$ of $GL(V)$, such that $\Theta_2(\rho^G_V(g))=\rho^G_V(\Theta(g))$, again by Theorem \ref{mos}. Thus since $\Theta_1$ and $\Theta_2$ are conjugated in $GL(V)$, then $\Theta_2=Ad(g)\circ\Theta_1\circ Ad(g^{-1})$ for some $g\in GL(V)$, hence $Ad(g)(\theta):=\theta'$ is an involution that will satisfy: $$\theta'\circ\rho^G_V(g)\circ\theta'=\rho^G_V(\Theta(g)), \forall g\in G,$$ and so $\rho^G_V$ is balanced as required. 	     \end{proof}

\section{Compatible representations}

\begin{defn} \label{triple} Let $G\subset G^{\mathbb{C}}\supset \tilde{G}$ be real forms, and $G\xrightarrow{\rho^G_V} GL(V)$ and $\tilde{G}\xrightarrow{\rho^{\tilde{G}}_{\tilde{V}}} GL(\tilde{V})$ be real representations of Lie groups. Suppose $G^{\mathbb{C}}\xrightarrow{\rho^{\mathbb{C}}}GL(V^{\mathbb{C}})$ is a complexified action of both $\rho^G_V$ and $\rho^{\tilde{G}}_{\tilde{V}}$. Then we say that $\rho^G_V$ is \textsl{compatible} with $\rho^{\tilde{G}}_{\tilde{V}}$, if the following two criterions are fulfilled:
\begin{enumerate}

\item{} $G$ and $\tilde{G}$ are compatible real forms of $G^{\mathbb{C}}$.
\item{} $V$ and $\tilde{V}$ are compatible real forms of $V^{\mathbb{C}}$.
\end{enumerate}
\end{defn}

Note that a real representation $G\rightarrow GL(V)$ with a complexification is always compatible with itself, and moreover if $U\subset G^{\mathbb{C}}$ is a compact real form, then a real Lie group action: $U\rightarrow GL(W)$, can always be complexified to a complex action: $G^{\mathbb{C}}\rightarrow GL(W^{\mathbb{C}})$, simply because $G^{\mathbb{C}}$ is the universal complexification group of $U$.

\begin{defn} \label{comp} Let $\rho^G_V, \rho^{\tilde{G}}_{\tilde{V}}$ and $\rho^U_W$ be pairwise compatible representations, where $U\subset G^{\mathbb{C}}$, is a compact real form. Then the triple: $\Big{(}\rho^G_V, \rho^{\tilde{G}}_{\tilde{V}}, \rho^U_W\Big{)}$ is said to be a \textsl{compatible triple}. \end{defn}

\begin{rem} When considering a compatible triple: $\Big{(} \mathfrak{g}, \tilde{\mathfrak{g}}, \mathfrak{u} \Big{)}$ of Lie algebras, there is a natural good choice of Cartan involutions, indeed the conjugation map $\tau$ of $\mathfrak{u}$, restricts to Cartan involutions: $\theta:=\tau_{|_{\mathfrak{g}}}$ and $\tilde{\theta}:=\tau_{|_{\tilde{\mathfrak{g}}}}$. In this way the global Cartan involutions of our groups $G=Ke^{\mathfrak{p}}$ and $\tilde{G}=\tilde{K}e^{\tilde{\mathfrak{p}}}$ are such that $K\subset U\supset \tilde{K}$, where $G^{\mb{C}}=Ue^{i\mathfrak{u}}$ is the global Cartan involution of $G^{\mb{C}}$, where $U$ has Lie algebra $\mathfrak{u}$, see Corollary \ref{aab}.  \end{rem}

From (\cite{W2}, Proposition A.2), a compatible pair: $\Big{(}\rho^G_V, \rho^{\tilde{G}}_{\tilde{V}}\Big{)}$ was considered. We now extend this result for compatible triples. We recall that an Hermitian inner product $H(-,-)$ on $V^{\mb{C}}$ which is real on a real subspace $V'\subset V^{\mb{C}}$ is said to be \emph{compatible} with $V'$.

\begin{lem} \label{b} Suppose $\Big{(}\rho^G_V, \rho^{\tilde{G}}_{\tilde{V}}, \rho^U_W\Big{)}$ is a compatible triple. Then there exist a $U$-invariant Hermitian inner product $H(-,-)$ on $V^{\mathbb{C}}$ which is compatible with $V, \tilde{V}$ and $W$. \end{lem}
\begin{proof} Since $U$ is compact then so is $\rho^{\mb{C}}(U)\subset (GL(V^{\mb{C}}))_{\mb{R}}\cong GL((V^{\mb{C}})_{\mb{R}}$. Set $E:=(V^{\mb{C}})_{\mb{R}}$ for the real vector space with complex structure $J$. Then the complex structure on $E$: $E\xrightarrow{J} E$ is an element of $GL(E)$, and so are all the conjugation maps: $\sigma_V, \tilde{\sigma}_{\tilde{V}}$ and $\tau_W$. Define the subgroup $U^*:=\langle \rho^{\mb{C}}(U), J, \sigma_V,\tilde{\sigma}_{\tilde{V}}, \tau_W \rangle\subset GL(E)$ then $U^*$ is a compact subgroup of $GL(E)$ since $\rho^{\mb{C}}(U)\subset U^*$ is closed and the quotient group $\frac{U^*}{\rho^{\mb{C}}(U)}$ is finite, using that $\Big{(}V,\tilde{V},W\Big{)}$ is a compatible triple. Now by the compatibility conditions on the Lie algebras we have that: $K\subset U\supset \tilde{K}$, thus $\rho^{\mb{C}}(K)\subset \rho^{\mb{C}}(U)\supset \rho^{\mb{C}}(\tilde{K})$. The inclusion $\phi$: $U^*\hookrightarrow GL(E),$ is a real representation of a compact Lie group. So there exist a $U^*$-invariant inner product $\langle-,-\rangle$ on $E$. Since $\langle-,-\rangle$ is $J$-invariant then it is easy to see that there exist a unique Hermitian inner product $H(-,-)$ on $V^{\mb{C}}$ with real part $\langle-,-\rangle$ on $E$. It is easy to check that $H(-,-)$ is $U$-invariant and therefore: $d\rho^{\mb{C}}(i\mathfrak{u})$ consists of Hermitian operators on $H(-,-)$. Also $H(-,-)$ is clearly $\Big{(}V,\tilde{V},W\Big{)}$-compatible by construction. The lemma is proved.

   \end{proof}

We thus also have an extended version of (\cite{W2}, Corollary A.2), concerning minimal vectors, which is essentially (\cite{RS}, Lemma 8.1) applied to each real representation:

\begin{cor}\label{m} Suppose $\Big{(}\rho^G_V, \rho^{\tilde{G}}_{\tilde{V}}, \rho^U_W\Big{)}$ is a compatible triple. Then there is a $\Big{(}V,\tilde{V}, W\Big{)}$-compatible $U$-invariant Hermitian inner product $H(-,-)$ on $V^{\mb{C}}$ such that: $$\mathcal{M}(U,W)\cup\mathcal{M}(\tilde{G},\tilde{V})\cup\mathcal{M}(G,V)\subset \mathcal{M}(G^{\mathbb{C}},V^{\mathbb{C}}).  $$  \end{cor}

Note that $\mathcal{M}(U,W)=W$, since $U$ is a compact real form. Now it follows from Proposition \ref{c}, that a compatible triple must be a balanced triple, i.e every real representation in the triple must be balanced:

\begin{cor} \label{v} Let $\Big{(}\rho^G_V, \rho^{\tilde{G}}_{\tilde{V}}, \rho^U_W\Big{)}$ be a compatible triple. Then there exist an involution: $\tau$, balancing $\rho^{\mathbb{C}}$, such that $\tau(V)\subset V$, $\tau(\tilde{V})\subset \tilde{V}$ and $\tau_W=1_W$. Thus $\rho^G_V$ and $\rho^{\tilde{G}}_{\tilde{V}}$ must also be balanced, with involutions: $\theta:=\tau_{V}$ and $\tilde{\theta}:=\tau_{\tilde{V}}$ respectively.  \end{cor}
\begin{proof} By Proposition \ref{c}, the conjugation map $\tau$ with fix points: $W$, which balance $\rho^{\mathbb{C}}$. Now since $W$ is pairwise compatible with $V$ and $\tilde{V}$, then obviously $\tau$ leaves $V$ and $\tilde{V}$ invariant. Thus since the global Cartan involution of $(G^{\mathbb{C}})_{\mathbb{R}}$ w.r.t the compact real form $U$ restricts to global Cartan involutions of $G$ and $\tilde{G}$ respectively, then obviously $\theta:=\tau_V$ and $\tilde{\theta}:=\tau_{\tilde{V}}$ balance $\rho^G_V$ and $\rho^{\tilde{G}}_{\tilde{V}}$ respectively. The corollary follows. \end{proof}

\begin{rem} We note in the proof of Lemma \ref{b}, that the $U$-invariant Hermitian inner product on $V^{\mathbb{C}}$ may be chosen to be invariant under $\tau$ from Corollary \ref{v}.  \end{rem}

We have the following criterion for a vector to be a minimal vector w.r.t a balanced Cartan involution: $\theta$: 

\begin{lem} [\cite{RS}, Lemma 5.1.1] Let $G\rightarrow GL(V)$ be a balanced real representation, and $\theta$ be an inner Cartan involution. Let $v=v_++v_-\in V$ be the Cartan decomposition, then $v\in \mathcal{M}(G,V)$ if and only if $\langle x\cdot v_+,v_- \rangle=0$ for all $x\in \mathfrak{p}$, where $\mathfrak{g}=\mathfrak{t}\oplus \mathfrak{p}$ is the Cartan decomposition of $\mathfrak{g}$ for which $\theta$ is balanced. \end{lem}

In particular we see that if $V=V_+\oplus V_-$ w.r.t $\theta$, then $V_+\cup V_-\subseteq \mathcal{M}(G,V)$. There are cases where $V_+\cup V_-=\mathcal{M}(G,V)$, for example the adjoint action of $SL_2(\mathbb{R})$ on $\mathfrak{sl}_2(\mathbb{R})$ or the matrix action of $O(p,q)$ on $\mathbb{R}^n$ with $n=p+q$. 
\\

Now for a compatible triple: $\Big{(}\rho^G_V, \rho^{\tilde{G}}_{\tilde{V}}, \rho^U_W\Big{)}$, where $V^{\mathbb{C}}=W\oplus iW$ w.r.t $\tau$ from Corollary \ref{v}, and $H(-,\tau(-))$ is a $U$-invariant Hermitian inner product compatible with $V,\tilde{V}$ and $W$, then we can characterise the minimal vectors as follows: 
\begin{enumerate} 
\item{} $\mathcal{M}(G,V)=\{v\in V|H(x\cdot v_+,v_-)=0, \forall x\in \mathfrak{p}\subset i\mathfrak{u}\}$. 
\item{} $\mathcal{M}(\tilde{G},\tilde{V})=\{\tilde{v}\in \tilde{V}|H(x\cdot \tilde{v}_+,\tilde{v}_-)=0, \forall x\in \tilde{\mathfrak{p}}\subset i\mathfrak{u}\}$. 
\item{} $\mathcal{M}(U,W)=W$. 
\item{} $\mathcal{M}(G^{\mathbb{C}},V^{\mathbb{C}})=\{v\in V^{\mathbb{C}}|H(x\cdot w_1,iw_2)=0, \forall x\in i\mathfrak{u}\}.$
\end{enumerate}

\section{Compatible real orbits}

\begin{defn} \label{aa} Let $\Big{(}\rho^G_V, \rho^{\tilde{G}}_{\tilde{V}}\Big{)}$ be a compatible pair. Suppose $v\in V$ and $\tilde{v}\in \tilde{V}$ are such that $\tilde{v}\in G^{\mathbb{C}}v$, then we shall say that $Gv$ is \textsl{compatible} with $\tilde{G}\tilde{v}$. \end{defn}

We shall write $Gv\sim \tilde{G}\tilde{v}$ for two compatible real orbits. One notes that if $U$ is compact, then by (\cite{RS}): $Uv_1\sim Uv_2$ if and only if $Uv_1=Uv_2$, this is however not true for general groups, see for example the adjoint action of $SL_2(\mathbb{R})$ on $\mathfrak{sl}_2(\mathbb{R})$. 

\begin{thm} \label{closed} Let $\Big{(}\rho^G_V, \rho^{\tilde{G}}_{\tilde{V}}, \rho^U_W\Big{)}$ be a compatible triple. Suppose $v\in V$ and $\tilde{v}\in \tilde{V}$ are such that $\tilde{G}\tilde{v}\sim Gv$. Assume $G^{\mathbb{C}}v\subset V^{\mathbb{C}}$ is closed. Then there exist inner Cartan involutions $\theta$ and $\tilde{\theta}$ of $V$ and $\tilde{V}$ respectively, such that if $v=v_++v_-$ and $\tilde{v}=\tilde{v}_++\tilde{v}_-$ are the Cartan decompositions, then: $$Gv_+\sim \tilde{G}\tilde{v}_+, \ and \ Gv_-\sim\tilde{G}\tilde{v}_-.$$ \end{thm}
\begin{proof} Since $G^{\mathbb{C}}v\subset V^{\mathbb{C}}$, is closed, then so are the real orbits: $Gv\subset V$, and $\tilde{G}\tilde{v}\subset \tilde{V}$ by Proposition \ref{linearly}, thus we can choose minimal vectors $X\in Gv$ and $\tilde{X}\in \tilde{G}\tilde{v}$. Now since $X$ and $\tilde{X}$ are also minimal vectors in $G^{\mathbb{C}}v$, then $\tilde{X}\in U\cdot X$ (by Corollary \ref{m}). So $X$ and $\tilde{X}$ have components which lie in the same $G^{\mathbb{C}}$-orbit, this follows since the $U$-action preserves the $W$-components and $iW$-components. But there exist $g\in G$ and $\tilde{g}\in \tilde{G}$, such that $g\cdot v=X$ and $\tilde{g}\cdot \tilde{v}=\tilde{X}$. So by conjugating our fixed inner Cartan involution of $\rho^G_V$ by the action of $g$, and similarly for $\rho^{\tilde{G}}_{\tilde{V}}$ by the action of $\tilde{g}$ we obtain the result. The theorem is proved.  \end{proof}

Following the proof of the theorem, then an interesting corollary is the following:

\begin{cor} \label{k} Let $\Big{(}\rho^G_V, \rho^{\tilde{G}}_{\tilde{V}}, \rho^U_W\Big{)}$ be a compatible triple. Suppose $v\in V$ and $\tilde{v}\in \tilde{V}$ are such that: $\tilde{G}\tilde{v}\sim Gv$. Then $Gv\cap V_+\neq \emptyset$ (respectively $Gv\cap V_-\neq\emptyset$) if and only if $\tilde{G}\tilde{v}\cap \tilde{V}_+\neq\emptyset$ (respectively $\tilde{G}\tilde{v}\cap \tilde{V}_-\neq \emptyset$). \end{cor}
\begin{proof} If $v_+\in Gv\cap V_+$ then as $V_+\subset \mathcal{M}(G,V)$, the real orbit: $Gv\subset V$, must be closed. Thus $\tilde{G}\tilde{v}\subset \tilde{V}$ must also be closed, and so we may choose a minimal vector $\beta\in \tilde{G}\tilde{v}$. But since $v_+\in W$ and $Uv_+\subset W$, because $U$ acts on $W$, then by Lemma \ref{m}, we have $\beta\in Uv_+\subset W$, thus $\beta\in \tilde{V}\cap W=\tilde{V}_+$. The other case is identical, since $U\cdot iW\subset iW$. The corollary is proved.  \end{proof}

Thus by letting $\tilde{G}:=U$ and $\tilde{V}:=W$ and $\rho^{\tilde{G}}_{\tilde{V}}:=\rho^U_W$ then: $\Big{(}\rho^G_V, \rho^U_W, \rho^U_W\Big{)}$ is a compatible triple and we get a new version of (\cite{W2}, Theorem 5.5 (case 2)), in view of inner Cartan involutions of the action: 

\begin{thm} \label{t} Let $(\rho^G_{V}, \rho^U_{W})$ be a compatible pair, then the following two statements hold:
\begin{enumerate}
\item{} Let $v\in V$, then the following statements are equivalent: 

\begin{enumerate}[label=\Alph*]
\item{}There exist $w\in W$ such that $Uw\sim Gv$. 
\item{}There exist an inner Cartan involution $V\xrightarrow{\theta} V$ such that $\theta(v)=v$.
\item{}There exist $w\in W$ such that $Uw\cap Gv\neq \emptyset$. 
\end{enumerate}

\item{} Let $v\in V$, then the following statements are equivalent: 
\begin{enumerate}[label=\Alph*]
\item{}There exist $iw\in iW$ such that $U\cdot iw\sim Gv$. 
\item{}There exist an inner Cartan involution $V\xrightarrow{\theta} V$ such that $\theta(v)=-v$. 
\item{}There exist $iw\in iW$ such that $U\cdot iw\cap Gv\neq \emptyset$. 
\end{enumerate}
\end{enumerate} \end{thm}
\begin{proof} We prove case (1) as case (2) is identical. ($A\Rightarrow B$).  Let $v\in V$ and write $v=v_++v_-$ w.r.t our inner Cartan involution: $\theta$. If there exist $w\in W$ such that $Gv\sim Uw$, then by Theorem \ref{closed}, it follows that $Gv_-\sim U w_-=\{0\}$, since the inner Cartan involution of $\rho^U_W$ is just the identity, and thus there exist $g\in G$ such that $g\cdot v\in Uw$, i.e $\theta(g\cdot v)=g\cdot v$. Therefore by conjugating $\theta$ by the action of $g$, we get a new inner Cartan involution $\theta'$, such that $\theta'(v)=v$. ($B\Rightarrow C$). Now if $\theta'(v)=v$ for some inner Cartan involution, then since $\theta'$ is conjugated to $\theta$ by definition, then it follows that there exist $g\in G$ such that $\theta(g\cdot v)=g\cdot v$, i.e $g\cdot v\in V_+\subset W$, and thus $Gv\cap U\cdot (g\cdot v)\neq \emptyset$, but $Uw=U\cdot (g\cdot v)$. ($C\Rightarrow A$). Finally if $v'\in Gv\cap Uw$ then clearly $Gv\sim Uv'$ for $v'\in W$. Thus the equivalences are established, and so the theorem is proved.  \end{proof}

Observe that the equivalence $A\Leftrightarrow C$ of case (1) is precisely Theorem \ref{W1r}. Now combining Corollary \ref{k} with Theorem \ref{t} we get the following invariance result of compatible real orbits:

\begin{cor}\label{q} Let $\Big{(}\rho^G_V, \rho^{\tilde{G}}_{\tilde{V}}, \rho^U_W\Big{)}$ be a compatible triple. Suppose $v\in V$ and $\tilde{v}\in\tilde{V}$ are such that $Gv\sim \tilde{G}\tilde{v}$. Then there exist an inner Cartan involution $V\xrightarrow{\theta} V$ such that $\theta(v)=v$ (respectively $\theta(v)=-v$) if and only if there exist an inner Cartan involution $\tilde{V}\xrightarrow{\tilde{\theta}}\tilde{V}$ such that $\tilde{\theta}(\tilde{v})=\tilde{v}$ (respectively $\tilde{\theta}(\tilde{v})=-\tilde{v}$).   \end{cor}
\begin{proof} It is enough to consider the case where $\theta(v)=v$. If $\theta$ is an inner Cartan involution of $V$ such that $\theta(v)=v$, then by Theorem \ref{t} case (1), $Gv\cap Uw\neq \emptyset$ for some $w\in W$. Thus the minimal vectors of $G^{\mathbb{C}}v$ is just $Uw\subset W$. In particular $\tilde{G}\tilde{v}$ must be closed as well, and thus $\tilde{G}\tilde{v}\cap Uw\neq \emptyset$, so we can choose an inner Cartan involution $\tilde{\theta}$ of $\rho^{\tilde{G}}_{\tilde{V}}$ such that $\tilde{\theta}(\tilde{v})=\tilde{v}$. The converse is identical, and so the corollary is proved.  \end{proof}

\begin{cor} Let $(\rho^G_{V}, \rho^U_{W})$ be a compatible pair. Let $v_1\in V$, and $G^{\mathbb{C}}v_1\cap V=Gv_1\cup Gv_2\cup\cdots\cup Gv_k$ for some natural number $k\geq 1$. Then there exist an inner Cartan involution $\theta_j$ of $\rho^{G}_V$ for some $1\leq j\leq k$ such that $\theta_j(v_j)=v_j$ (respectively $\theta_j(v_j)=-v_j$) if and only if there exist inner Cartan involutions: $\theta_i$ for all $1\leq i\leq k$ such that $\theta_i(v_i)=v_i$ (respectively $\theta_i(v_i)=-v_i$).   $\square$  \end{cor}

For non-closed orbits we can also apply Theorem \ref{closed} to their boundaries:

\begin{cor} \label{end} Suppose we have compatible triple: $\Big{(}\rho^G_V, \rho^{\tilde{G}}_{\tilde{V}}, \rho^U_W\Big{)}$.  Let $v\in V$ and $\tilde{v}\in \tilde{V}$. Assume $G^{\mb{C}}v$ is not closed, and $Gv\sim \tilde{G}\tilde{v}$. Let $p\in \mathfrak{p}$ and $\tilde{p}\in \tilde{\mathfrak{p}}$ be such that the limits exist: $e^{tp}\cdot v\rightarrow \alpha\in \overline{Gv}-Gv$ and $e^{t\tilde{p}}\cdot \tilde{v}\rightarrow \tilde{\alpha}\in \overline{\tilde{G}\tilde{v}}-\tilde{G}\tilde{v}$ where $G\alpha$ and $\tilde{G}\tilde{\alpha}$ are closed (Theorem \ref{RS}). Then there exist inner Cartan involutions $\theta$ and $\tilde{\theta}$ of $V$ and $\tilde{V}$ respectively, such that if $\alpha=\alpha_++\alpha_-$ and $\tilde{\alpha}=\tilde{\alpha}_++\tilde{\alpha}_-$ are the Cartan decompositions, then: $$G\alpha_+\sim \tilde{G}\tilde{\alpha}_+, \ and \ G\alpha_-\sim\tilde{G}\tilde{\alpha}_-.$$  \end{cor}
\begin{proof} Since there is a unique closed $G^{\mb{C}}$-orbit in the closure $\overline{G^{\mb{C}}v}$ (Theorem \ref{RS}), and $Gv\sim \tilde{G}\tilde{v}$, then the real orbits in the closures must be compatible, i.e $G\alpha\sim \tilde{G}\tilde{\alpha}$. Thus we may apply Theorem \ref{closed}, and the corollary follows. \end{proof}

We end this section with an example illustrating the falsehood of Theorem \ref{t} in the case where both groups are non-compact:

\begin{ex} \textbf{[Not all compatible real orbits need to intersect]}. Let $G:=SL_2(\mathbb{R})\subset G^{\mb{C}}:=SL_2(\mb{C})\supset U:=SU(2)$ be the standard matrix representations, and consider the adjoint actions of these groups on their Lie algebras respectively. It is easy to see that $G$ is compatible with $U$. We can find $v'\neq v\in \mathfrak{g}\cap \mathfrak{u}$ such that $Gv\subset G^{\mb{C}}v\supset Gv'$ but $Gv\neq Gv'$. Consider the induced product action of the semi-simple groups: $$H:=G\times G\subset G^{\mb{C}}\times G^{\mb{C}}\supset G\times U:=\tilde{H},$$ acting on $\mathfrak{h}:=\mathfrak{g}\times \mathfrak{g}$ and $\tilde{\mathfrak{h}}:=\mathfrak{g}\times \mathfrak{u}$ respectively. Then $\Big{(} H, \tilde{H}, U\times U \Big{)}$ is a compatible triple, and $\Big{(} \mathfrak{h}, \tilde{\mathfrak{h}}, \mathfrak{u}\times\mathfrak{u} \Big{)}$ is also a compatible triple. Thus we have the setup of compatible representations. Now we note that: $$H\cdot (v,v)\subset H^{\mb{C}} (v,v)\supset \tilde{H}\cdot (v',v'),$$ however if there exist $(v_1,v_2)\in H\cdot (v,v)\cap \tilde{H}\cdot (v',v')$, then there exist $g\in G$ such that $g\cdot v=v'$ which is impossible. Hence $H\cdot (v,v)\sim \tilde{H}\cdot (v',v')$ are compatible real orbits, but cannot intersect.

 \end{ex}

\section{Applications to Wick-rotations of arbitrary signatures}

\subsection{The isometry action of $O(n,\mb{C})$ on tensors} \label{tensors}

In this subsection we consider Wick-rotations and recall the setup from \cite{W2}. We use the isometry action of the complex orthogonal group: $O(n,\mb{C})$ on a tensor product space, induced from the isometry action of the holomorphic metric, and apply the results of Section 5 to obtain necessary conditions for the existence of a Wick-rotation at a common fix point $p$. We begin by observing that we indeed have the setup of compatible representations (see Section 4).

\

Suppose now that $(M,g)\subset (M^{\mb{C}}, g^{\mb{C}})\supset (\tilde{M}, \tilde{g})$ are Wick-rotated at $p\in M\cap\tilde{M}$, and consider now the complex isometry action $\rho^{\mb{C}}$ of $O(n,\mathbb{C})$ on $T_pM^{\mathbb{C}}$: $$g\cdot v:=g(v), \ g\in O(n,\mb{C}), \ v\in T_pM^{\mb{C}}.$$ Now by using an isomorphism: $T_pM^{\mathbb{C}}\xrightarrow{\psi} T_pM^{\mathbb{C}}$, as in Corollary \ref{sw}, then we have a compatible triple: $\Big{(}T_pM, \psi^{-1}(T_p\tilde{M}), W\Big{)}$, and we know by Lemma \ref{op}, that the corresponding pseudo-orthogonal groups: $\Big{(}O(p,q),O(\tilde{p},\tilde{q}), O(n)\Big{)}$ also form a compatible triple (by definition). Thus the corresponding real isometry actions of our pseudo-inner products: $g(-,-), \ \tilde{g}(-,-)$ and $g^{\mb{C}}_{|_W}(-,-)$, are restrictions of $\rho^{\mb{C}}$. Denote them by $\rho^{O(p,q)}_{T_pM}$, $\rho^{O(\tilde{p},\tilde{q})}_{\psi^{-1}(T_p\tilde{M})}$ and $\rho^{O(n)}_{W}$ respectively, then they form a compatible triple: $\Big{(} \rho^{O(p,q)}_{T_pM}, \rho^{O(\tilde{p},\tilde{q})}_{\psi^{-1}(T_p\tilde{M})}, \rho^{O(n)}_{W}\Big{)}$, in the sense of Definition \ref{comp}. 

The map $\psi$ and the isometry action $\rho^{\mb{C}}$ naturally extends tensorially to complexified tensors: $v^{\mathbb{C}}\in \mathcal{V}^{\mathbb{C}}:=\Big{(}\bigotimes_{i=1}^k T_pM^{\mathbb{C}}\Big{)}\bigotimes\Big{(}\bigotimes_{i=1}^m ({T_pM^\mb{C}})^*\Big{)}$ at a point $p$. Denote $\Psi$ for the extension map of $\psi$ to tensors $\mathcal{V}^{\mb{C}}$, i.e $\Psi(-):=\psi\cdot (-)$. Then it is easy to check that the triple: $\Big{(}\mathcal{V}, \Psi^{-1}(\tilde{\mathcal{V}}), \mathcal{W}\Big{)}$ also form a compatible triple, where we define: 

$$\mathcal{V}:=\Big{(}\bigotimes_{i=1}^k T_pM\Big{)}\bigotimes\Big{(}\bigotimes_{i=1} ^m (T_pM)^*\Big{)}, \ \ \tilde{\mathcal{V}}:=\Big{(}\bigotimes_{i=1}^k T_p\tilde{M}\Big{)}\bigotimes\Big{(}\bigotimes_{i=1} ^m (T_p\tilde{M})^*\Big{)},$$ and $\mathcal{W}:=\Big{(}\bigotimes_{i=1}^k W\Big{)}\bigotimes\Big{(}\bigotimes_{i=1}^m W^*\Big{)}.$

Thus the real isometry tensor actions also naturally form a compatible triple: $\Big{(} \rho^{O(p,q)}_{\mathcal{V}}, \rho^{O(\tilde{p},\tilde{q})}_{\Psi^{-1}(\tilde{\mathcal{V}})},  \rho^{O(n)}_{\mathcal{W}}\Big{)}$. 

Let $\{e_1,\dots, e_p, \dots, e_n\}$ be a pseudo-orthonormal basis of the metric $g$, and $\theta$ the Cartan involution w.r.t this basis. Then $$\{y_1,\dots, y_n\}:=\{e_1,\dots, e_p, ie_{p+1}, \dots, ie_n\},$$ is an orthonormal basis of $g^{\mb{C}}$. Note that the span of $\{y_1,\dots, y_n\}$ is precisely the compact real slice $W$, and moreover the conjugation map $\tau$ of $W$ in $T_pM^{\mb{C}}$ restricts to $\theta$. We can extend the holomorphic metric $g^{\mb{C}}$ (at $p$) to a holomorphic inner product $\textbf{g}^{\mb{C}}$ on $\mathcal{V}^{\mb{C}}$ by defining: $$\textbf{g}^{\mb{C}}\Big{(}\otimes_{i=1}^{k} v^{\mb{C}}_i\otimes_{j=1}^{m} w^{\mb{C}^*}_j, \otimes_{s=1}^{k} \tilde{v}^{\mb{C}}_s\otimes_{t=1}^{m} \tilde{w}^{\mb{C}^*}_t\Big{)}:=\sum_{1\leq i,s\leq n} g^{\mb{C}}(v^{\mb{C}}_i, \tilde{v}^{\mb{C}}_s)+\sum_{1\leq j,t\leq n}g^{\mb{C}}(w^{\mb{C}}_j, \tilde{w}^{\mb{C}}_t),$$ using the isomorphism: $$T_pM^{\mb{C}}\xrightarrow{{v^{\mb{C}}}^*} {(T_pM^{\mb{C}})}^*, \ v\mapsto g^{\mb{C}}(v,-).$$ 

We see that $\mathcal{V}\subset (\mathcal{V}^{\mb{C}}, \textbf{g}^{\mb{C}})\supset \Psi^{-1}(\tilde{\mathcal{V}})$ are real forms (i.e real slices). Denote $\textbf{g}$ for the induced pseudo-inner product on $\mathcal{V}$. The Cartan involution $\theta$ of $g$ extends in the obvious way to a Cartan involution $\Theta$ of $\textbf{g}$, by $$\otimes_{i=1}^{k} v_i\otimes_{j=1}^{m} v^*_j \mapsto \otimes_{i=1}^{k} \theta(v_i)\otimes_{j=1}^{m} v^*_j\circ\theta,$$ which is just the action of $\theta$ on tensors, i.e $\Theta=\rho^{O(p,q)}_{\mathcal{V}}(\theta)(v)$. Now the inner Cartan involutions of the action (w.r.t \textbf{g}) are just those conjugate to $\Theta$ by definition (see definition in Section 3). This means that the inner Cartan involutions are precisely those which are extensions from a Cartan involution of the metric $g$. 

Moreover because $T_pM$ and $\psi^{-1}(T_p\tilde{M})$ are both compatible with $W$, then we also have that $\mathcal{V}$ and $\Psi^{-1}(\tilde{\mathcal{V}})$ are compatible with the $O(n)$-invariant Hermitian inner product: $\textbf{g}^{\mb{C}}(\cdot, \mathcal{T}(\cdot))$, where $\mathcal{T}$ is the conjugation map of $\mathcal{W}\subset \mathcal{V}^{\mb{C}}$ defined by the action: $\mathcal{T}(v^{\mb{C}}):=\tau\cdot v^{\mb{C}}$. Thus the isometry actions lend themselves to the results of Section 5.

\begin{rem} The isometry tensor product action and everything defined in this section extends in the natural way to finite sums of the form: $$\bigoplus_{k,m}\Big{(}\Big{(}\bigotimes_{i=1}^k T_pM^{\mathbb{C}}\Big{)}\bigotimes\Big{(}\bigotimes_{i=1}^m ({T_pM^\mb{C}})^*\Big{)}\Big{)}.$$ Thus from heron and to the end of this paper we assume the isometry tensor action on such sums and thus replace: $\mathcal{V}^{\mathbb{C}}$ with this sum. \end{rem}

\begin{defn} \label{h} Let $M$ and $\tilde{M}$ be two Wick-rotatable real slices at $p\in M\cap \tilde{M}$. Then two tensors $v\in \mathcal{V}$ and $\tilde{v}\in \tilde{\mathcal{V}}$ are said to be \textsl{Wick-rotatable} at $p$, if they lie in the same $O(n,\mathbb{C})$-orbit, i.e $$O(n,\mathbb{C})\cdot v=O(n,\mathbb{C})\cdot \tilde{v}.$$   \end{defn}

Note that if $v$ and $\tilde{v}$ are two Wick-rotatable tensors, then using the map $\Psi$ above, we see that $O(p,q)v\sim O(\tilde{p},\tilde{q})\Psi^{-1}(\tilde{v})$ are two compatible real orbits (see Definition \ref{aa}). 

The most obvious example of two Wick-rotatable tensors, are of course the real metrics themselves: $g\in T^2(T_pM)$ and $\tilde{g}\in T^2(T_p\tilde{M})$ at the common point $p$, simply because they are restrictions of the holomorphic metric at $p$. Thus from the metrics it follows that the real Levi-Civita connections: $\nabla \in T^2(T_pM)$ and $\tilde{\nabla}\in T^2(T_p\tilde{M}) $ must also be restrictions of the holomorphic Levi-Civita connection: $\nabla^{\mb{C}}$, on the tangent spaces at $p$. Thus furthermore the real Riemann tensors: $R$ and $\tilde{R}$ restricted to the tangent spaces at $p$ are also restrictions of the holomorphic Riemann tensor. As seen in \cite{W1}, one can for instance view them as vectors: $R\in End(\mathfrak{o}(p,q))$ and $\tilde{R}\in End(\mathfrak{o}(\tilde{p},\tilde{q}))$. From the Riemann tensors it also follows that the real Ricci curvatures: $ric_g\in T^2(T_pM)$ and $ric_{\tilde{g}}\in T^2(T_p\tilde{M})$ and the real Ricci operators: $Ric_g\in End(T_pM)$ and $Ric_{\tilde{g}}\in End(T_p\tilde{M})$ must also be Wick-rotatable respectively.

\subsection{Purely electric/magnetic spaces} \label{purely}
Let $(M,g)$ be a pseudo-Riemannian space of signature $(p,q)$, and let $p\in M$ be a point, and $\theta\in O(p,q)$ be a Cartan involution of $g_p(-,-)$. Consider the isometry tensor action of $O(p,q)$ on $\mathcal{V}$ from the previous section: $$O(p,q)\xrightarrow{\rho^{O(p,q)}_{\mathcal{V}}} GL(\mathcal{V}).$$ Then $\theta$ naturally extends to an involution $\Theta:=\rho^{O(p,q)}_{\mathcal{V}}(\theta)$ on $\mathcal{V}$, and the metric naturally induces a pseudo-inner product: $\textbf{g}(-,-)$  on $\mathcal{V}$ such that $\Theta$ is a Cartan involution. 

Let now $R\in \mathcal{V}$ be the Riemann tensor of $M$ at $p$ for $\mathcal{V}$ some tensor product. If there exist a Cartan involution $\Theta$ such that $\Theta(R)=R$ (respectively $\Theta(R)=-R$), then the space $(M,g)$ at $p$ is called \textsl{Riemann purely electric} (RPE) (respectively \textsl{Riemann purely magnetic} (RPM)). If there is such a $\Theta$ for the Weyl tensor at $p$, then $(M,g)$ at $p$ is called \textsl{purely electric} (PE) (respectively \textsl{purely magnetic} (PM)).  

\subsection{Invariance theorem for Wick-rotation at a point $p$}

We now follow the notation of Section \ref{tensors} for the isometry action on tensor products, and apply the results of Section 5 to these actions. For the results in this section and the next we can for instance consider the Wick-rotatable tensors mentioned in the last paragraph after Defn \ref{h}. Recall the result given in \cite{W1}, where a Wick-rotation of a Riemannian real slice and an arbitrary pseudo-Riemannian real slice was considered. There the following result was proven:

\begin{thm} [\cite{W1}] \label{HH} Let $(M,g)$ and $(\tilde{M},\tilde{g})$ be Wick-rotated at $p\in M\cap\tilde{M}$. Assume $(\tilde{M},\tilde{g})$ is Riemannian. Then the pseudo-Riemannian space $(M,g)$ is Riemann purely electric (RPE) at $p$.  \end{thm}

We note in the case where $(\tilde{M},\tilde{g})$ is Riemannian, then the complex orbit: $O(n,\mb{C})v\subset \mathcal{V}^{\mb{C}}$, for two Wick-rotatable tensors is always closed. Moreover any Cartan involution for a Riemannian space is just the identity ($\theta=1$), and thus when extended to tensors, this is just the identity as well ($\Theta=1$).

Thus for arbitrary signatures the following result is a generalisation:

\begin{thm}\label{con} Let $(M,g)$ and $(\tilde{M},\tilde{g})$ be Wick-rotated at $p\in M\cap\tilde{M}$ of arbitrary signatures. Suppose $v\in \mathcal{V}$ and $\tilde{v}\in \tilde{\mathcal{V}}$ are two Wick-rotated tensors at $p$. Assume $O(n,\mb{C})v$ is closed. Then there exist Cartan involutions $\theta$ and $\tilde{\theta}$ of $g(-,-)$ and $\tilde{g}(-,-)$ at $p$ respectively, such that if $v=v_++v_-$ and $\tilde{v}=\tilde{v}_++\tilde{v}_-$ are the Cartan decompositions w.r.t the extended Cartan involutions on $\mathcal{V}$ and $\tilde{\mathcal{V}}$, then $v_+$ and $\tilde{v}_+$ are Wick-rotated at $p$ and so are $v_-$ and $\tilde{v}_-$ at $p$.   \end{thm}
\begin{proof} We apply Theorem \ref{closed} to the compatible triple: 
\\
 $\Big{(} \rho^{O(p,q)}_{\mathcal{V}}, \rho^{O(\tilde{p},\tilde{q})}_{\Psi^{-1}(\tilde{\mathcal{V}})},  \rho^{O(n)}_{\mathcal{W}}\Big{)}$, as defined in Section \ref{tensors}, together with    the compatible real orbits: $O(p,q)v\sim O(\tilde{p},\tilde{q})\Psi^{-1}(\tilde{v})$. The results then follow to $\tilde{v}$, as $\Psi(-)=g\cdot -$ for some $g\in O(n,\mb{C})$.\end{proof} 

Thus from the theorem there are Cartan involutions such that the components must be Wick-rotated also at $p$. Note that Theorem \ref{t} (case 1), is precisely Theorem \ref{HH}. Recall now the definition given in Section \ref{purely}, then we have the following invariance result for Wick-rotation at a point which also extends Theorem \ref{HH}:

\begin{cor} \textbf{[Invariance of Wick-rotation]}. Let $(M,g)$ and $(\tilde{M},\tilde{g})$ be Wick-rotated at $p\in M\cap\tilde{M}$ of arbitrary signatures. Then $M$ is (PE), (RPE), (PM) or (RPM) if and only if $\tilde{M}$ is (PE), (RPE), (PM) or (RPM) respectively.  \end{cor}
\begin{proof} This is precisely Corollary \ref{q} with $v=R$ (respectively $v=W$) and $\tilde{v}:=\tilde{R}$ (respectively $\tilde{v}=\tilde{W}$) being the Riemann tensors at $p$ (respectively the Weyl tensors at $p$), applied to the compatible triple: 
\\
 $\Big{(} \rho^{O(p,q)}_{\mathcal{V}}, \rho^{O(\tilde{p},\tilde{q})}_{\Psi^{-1}(\tilde{\mathcal{V}})},  \rho^{O(n)}_{\mathcal{W}}\Big{)}$, as defined in Section \ref{tensors}.    \end{proof}

One may conjecture that Theorem \ref{con} also hold for non-closed orbits: $O(n,\mb{C})v$, so this is a natural follow-up question to ask. However for a non-closed orbit: $O(n,\mb{C})v$, we do have the following result on the boundaries of the orbits:

\begin{cor} Let $(M,g)$ and $(\tilde{M},\tilde{g})$ be Wick-rotated at $p\in M\cap\tilde{M}$ of arbitrary signatures. Let $v\in \mathcal{V}$ and $\tilde{v}\in \tilde{\mathcal{V}}$. Assume $O(n,\mb{C})v$ is not closed, and that $v$ is Wick-rotatable to $\tilde{v}$ at $p$. Let $x\in \mathfrak{p}$ and $\tilde{x}\in \tilde{\mathfrak{p}}$ be such that the limits exist: $e^{tx}\cdot v\rightarrow \alpha\in \overline{O(p,q)v}-O(p,q)v$ and $e^{t\tilde{x}}\cdot \tilde{v}\rightarrow \tilde{\alpha}\in \overline{O(\tilde{p},\tilde{q})\tilde{v}}-O(\tilde{p},\tilde{q})\tilde{v}$ where $O(p,q)\alpha$ and $O(\tilde{p},\tilde{q})\tilde{\alpha}$ are closed (Theorem \ref{RS}). Then there exist extended Cartan involutions $\Theta$ and $\tilde{\Theta}$ of $\mathcal{V}$ and $\tilde{\mathcal{V}}$ respectively, such that if $\alpha=\alpha_++\alpha_-$ and $\tilde{\alpha}=\tilde{\alpha}_++\tilde{\alpha}_-$ are the Cartan decompositions, then $\alpha_+$ is Wick-rotated at $p$ to $\tilde{\alpha}_+$ and $\alpha_-$ is Wick-rotated at $p$ to $\tilde{\alpha}_-$.  \end{cor}
\begin{proof} We apply Corollary \ref{end} to the compatible triple: 
\\
 $\Big{(} \rho^{O(p,q)}_{\mathcal{V}}, \rho^{O(\tilde{p},\tilde{q})}_{\Psi^{-1}(\tilde{\mathcal{V}})},  \rho^{O(n)}_{\mathcal{W}}\Big{)}$, as defined in Section \ref{tensors}, and to the compatible real orbits: $O(p,q)v\sim O(\tilde{p},\tilde{q})\Psi^{-1}(\tilde{v})$. \end{proof}

\subsection{A note on Wick-rotations of the same signatures}

Suppose $M$ and $\tilde{M}$ are Wick-rotated at a common point $p$ and have the same signatures. Then if $v\in \mathcal{V}$ and $\tilde{v}\in \tilde{\mathcal{V}}$ are Wick-rotated at $p$, we may choose $\Psi$ (in Definition \ref{h}) such that $\Psi^{-1}(\tilde{v})\in \mathcal{V}$. Thus we have the following:

\begin{prop} \label{a} Suppose $(M,g)$ and $(\tilde{M},\tilde{g})$ are Wick-rotated at $p\in M\cap\tilde{M}$ and have the same signatures $p+q=n$. Let $v\in\mathcal{V}$ and $\tilde{v}\in \tilde{\mathcal{V}}$ be two Wick-rotatable tensors at $p$. Assume there is a unique real orbit in the complex orbit $O(n,\mb{C})v$, i.e $O(n,\mb{C})v\cap \mathcal{V}=O(p,q)v$. Then there is a homeomorphism: $$\mathcal{V}\supset O(g_p(-,-),T_pM)\cdot v\cong O(\tilde{g}_p(-,-), T_p\tilde{M})\cdot \tilde{v}\subset \tilde{\mathcal{V}}.$$ \end{prop}
\begin{proof} Since $v\in\mathcal{V}$ and $\tilde{v}\in\tilde{\mathcal{V}}$ are Wick-rotatable tensors at $p$, then $\Psi^{-1}(\tilde{v})\in O(n,\mathbb{C})v$, where we may choose $\psi$ to be an isomorphism: $T_pM\cong T_p\tilde{M}$. Thus $\Psi^{-1}(\tilde{v})\in \mathcal{V}$, and so $$O(\tilde{g}_p(-,-), T_p\tilde{M})\cdot \tilde{v}\cong O(p,q)\cdot\Psi^{-1}(\tilde{v})=O(p,q)v:=O(g_p(-,-),T_pM)\cdot v.$$ \end{proof}

Thus for two Wick-rotated Riemannian real slices at a common point we have:

\begin{cor} Suppose $(M,g)$ and $(\tilde{M},\tilde{g})$ are Wick-rotated Riemannian slices at $p\in M\cap\tilde{M}$. Let $v\in\mathcal{V}$ and $\tilde{v}\in \tilde{\mathcal{V}}$ be two Wick-rotatable tensors at $p$. Then there is a diffeomorphism of embedded submanifolds: $$\mathcal{V}\supset O(g_p(-,-), T_pM))\cdot v\cong O(\tilde{g}_p(-,-), T_p\tilde{M})\cdot\tilde{v}\subset \tilde{\mathcal{V}}.$$ \end{cor}
\begin{proof} By (\cite{RS}, Proposition 8.3.1), we can apply Proposition \ref{a}, and the result follows. \end{proof}

\subsection{Wick-rotatable metrics} \label{Wickm}

Here we will consider two pseudo-Riemannian metrics $(M,g)$ and $(\tilde{M},\tilde{g})$ of possibly different signature and give sufficient conditions 
when such are Wick-rotated. 
\begin{prop}
Assume that $v\in V$ and $\tilde{v}\in\tilde{V}$ have closed orbits $Gv$ and $\tilde{G}\tilde{v}$, and that their polynomial invariants are identical. Then they are Wick-rotated in the sense that there is a $G^\C\supset G, \tilde{G}$ so that 
\[Gv\subset G^\C v\supset \tilde{G}\tilde{v} \]
\end{prop}
\begin{proof}
Since $v$ and $\tilde{v}$ have identical invariants and their corresponding orbits are closed, then due to Thm. \ref{clos}, then the corresponding complex orbits are closed too. Then, since the invariants separate the complex orbits, the complex orbits are identical and the result follows. 
\end{proof}
Let now $\mathcal{V}^{(k)}$ be the vector space associated with the components of tensors, ref. Section \ref{tensors}, so that $\bigoplus_{i=0}^k\nabla^{(i)}\text{Riem}\in \mathcal{V}^{(k)}$, where $\nabla^{(i)}\text{Riem}$ indicates the $i$th covariant derivative of the Riemann curvature tensor. Then: 
\begin{thm}
Let $v^k \in \mathcal{V}^{(k)}$ and $\tilde{v}^k\in \tilde{\mathcal{V}}^{(k)}$ be the curvature tensors of $(M,g)$ and $(\tilde{M},\tilde{g})$, respectively. Assume that there exists points $q\in M$ and $\tilde{q}\in\tilde{M}$ so that the corresponding orbits $Gv^k$ and $\tilde{G}\tilde{v}^k$ are closed and their invariants are identical for all $k$. Then the metrics are Wick-rotated w.r.t a common point $q=\tilde{q}$.
\end{thm}
\begin{proof}
By the above proposition,  $Gv^k$ and $\tilde{G}\tilde{v}^k$ are Wick-rotated for a $q\in M$ and $\tilde{q}\in \tilde{M}$,  for all $k$. Since the metrics are real analytic, there exists neighbourhoods $U\subset M$ and $\tilde{U}\subset \tilde{M}$, of $q$ and $\tilde{q}$ respectively, which can be embedded into a complex neighbourhood $U^{\C}$ so that $q=\tilde{q}\in U^\C$. The real analytic structure can now be extended to an analytic structure on $U^\C$ and the complexified orbit $G^\C v^k$ at $q=\tilde{q}$ give rise to complex curvature tensors. These can now be (maximally) analytically extended to an analytic metric $g^\C$ on a neighbourhood in $U^\C$ (for simplicity, call this neighbourhood $U^\C$). These real analytic structures $(U,g|_U)$ and $(\tilde{U},\tilde{g}|_{\tilde{U}})$ thus are Wick-rotated, both being restrictions of the complex holomorphic $(U^\C,g^\C)$. Since Wick-rotation is a local criterion, the theorem now follows. 
\end{proof}
This implies that for closed orbits, the metrics are necessarily Wick-rotated as long as their invariants are identical. If the orbits are not closed, we have a result which is point-wise. By evaluating the curvature tensors at a point, we can use the following result. 
\begin{thm}
Assume that 	$v\in V$ and $\tilde{v}\in\tilde{V}$ have identical invariants. Then there exist $p\in\mathfrak{p}$ and $\tilde{p}\in\tilde{\frak{p}}$ so that $v_0:=\lim_{t\rightarrow\infty}e^{tp}\cdot v$ and $\tilde{v}_0:=\lim_{t\rightarrow\infty}e^{t\tilde{p}}\cdot\tilde{v}$ are Wick-rotated. 
\end{thm}
\begin{proof}
This follows from Theorem \ref{RS}, and the fact that a point in the orbit, $x\in Gv$ and any point its closure $x_0\in\overline{Gv}$ have identical invariants. Since there is a unique closed orbit in the closure $\overline{Gv}$, the result follows. 
\end{proof}
Note that we say that a metric $(M,g)$ is characterised by its invariants is exactly when $v$ has a closed orbit. This implies that for two metrics being characterised by its invariants which have identical invariants are related by Wick rotations. 

\subsection{Universal metrics}
A pseudo-Riemannian metric is called \emph{universal} if  all conserved symmetric rank-2 tensors constructed
  from the metric, the Riemann tensor and its covariant derivatives are multiples of the metric.
Hence, universal metrics are metrics which obey $T_{\mu\nu}=\lambda g_{\mu\nu}$, for all symmetric conserved tensors $T_{\mu\nu}$ constructed from the metric and the curvature tensors (recall that conserved implies $\nabla^{\mu}T_{\mu\nu}=0$) \cite{CGHP}. We note that this constuction can be lifted holomorphically to the holomorphic Riemannian manifold and thus implies $T^\C_{\mu\nu}=\lambda g^\C_{\mu\nu}$. Thus, universality is \emph{preserved under Wick rotation}. This straight-forwardly leads to:
\begin{prop}
Assume that $(M,g)$ and $(\tilde{M},\tilde{g})$ are two Wick-rotated pseudo-Riemannian manifolds. Then $(M,g)$ is universal if and only if $(\tilde{M},\tilde{g})$ is universal. 
\end{prop}
In the Riemannian case, all such metrics are classified. Indeed, all Riemannian universal spaces are locally homogeneous space where the isotropy group acts irreducibly on the tangent space \cite{Bleecker}. In other signatures this is no longer true as there are universal examples of both Kundt and Walker type which are not locally homogeneous \cite{CGHP, universal}. It is, however, interesting to study those that are Wick-rotatable to the Riemannian case and relate these to the irreducibly-acting isotropy group. 

As an example, consider the following four-dimensional Riemannian metric, 
\beq
g=g_{S^2}\oplus g_{S^2},
\eeq
where $g_{S^2}$ is the unit metric on the sphere. This has an isotropy group $O(2)\times O(2)\times \mathbb{Z}_2$, where the $\mathbb{Z}_2$ interchanges the two spheres. Each of the two spheres can be Wick-rotated to other two-dimensional spaces of constant curvature:
\[ g_{S^2}\longmapsto g_{dS}~~ (-+), \qquad -g_{AdS}~~(+-), \qquad -g_{H^2}~~(--),\]
where $(A)dS$ is (anti-)de Sitter space, and $H^2$ is the unit hyperbolic space. These can now be combined in various ways to get various Wick-related spaces being universal. For example, 
\[ g_N=g_{S^2}\oplus (-g_{H^2})\] 
is a universal metric of neutral signature, and  
\[ g_L=g_{S^2}\oplus g_{dS}\] 
is a universal metric of Lorentzian signature. 
We note that the interchange symmetry $\mathbb{Z}_2$ of the Riemannian metric is not necessarily an isotropy of the Wick-related metrics. Indeed, in both examples above, there exist vectors $X\in T_pM$ so that $g_{\bullet}(X,X)>0$, while $g_\bullet(A(X),A(X))<0$, where $A(X)$ is the action of the non-trivial element of $\mathbb{Z}_2$ on $X$. Thus, the $\mathbb{Z}_2$ action cannot be an isotropy of $g_\bullet$. On the other hand, the symmetry $\mathbb{Z}_2$ preserves the signature and maps metrics onto other Wick-related metrics. 

\subsection{On the set of tensors with identical invariants}

Let $(M,g)$ be a pseudo-Riemannian manifold and denote $v^{(l)}:=\nabla^{(l)}\text{Riem}$ for the $l$th covariant derivative of the Riemann tensor at a fixed point $p\in M$. Define $V^{(l)}$ to be a tensor product space for which $v^{(l)}\in V^{(l)}$. Set $\mathcal{V}^{(k)}:=\bigoplus_{l=0}^k V^{(l)}$, then it contains all the covariant derivatives $v^{(l)}$, up to order $k$ of the Riemann tensor at the fixed point $p$. The isometry group $O(p,q)$ of the pseudo-inner product $g(-,-)$ (at $p$) acts on $\mathcal{V}^{(k)}$ by the tensor product action (as defined in Section \ref{tensors}). Consider the algebra of polynomial invariants $\mathbb{R}[\mathcal{V}^{(k)}]^{O(p,q)}$ of the action. Let $I$ be the polynomial invariants restricted to the set of all the covariant derivatives of the Riemann tensor. Then $I$ is defined to be the set of \emph{polynomial curvature invariants}. Moreover let $I_k$ denote the polynomial invariants $\mathbb{R}[\mathcal{V}^{(k)}]^{O(p,q)}$ restricted to the set of all the $v^{(l)}$ up to $k$th order. Moreover, $I=I_k$ is finitely generated \cite{GW}, which means that we can find a finite number of generators for $I$: $I=\langle f_1,f_2,\dots, f_N\rangle$. Set $\mathcal{V}:=\mathcal{V}^{(k)}$ then the set of invariants $I$ defines a polynomial function: $$\mathfrak{I}: ~\mathcal{V}\longrightarrow \R^N, \ \ v\mapsto (f_1(v),\dots, f_N(v)).$$ We recall that the space $(M,g)$ is said to be VSI if $I=\{0\}$, and is said to be $VSI_k$ if $I_k=0$.

Let $x_p\in\mathcal{V}$, and consider the set $S:=\mathfrak{I}^{-1}(\mathfrak{I}(x_p))\subset \mathcal{V}$, which is the set of all tensors in $\mathcal{V}$ having identical invariants as $x_p$. 

Let $S_i$ be the connected components of $S$ so that $S=\cup_{i=1}^nS_i$, and $S_i\cap S_j=\emptyset,~ i\neq j$. 
Then, regarding  the topology of the set $S$:
\begin{prop} Let $S$ be as above. Then:
\begin{enumerate}
\item{} If $\mathfrak{I}(S)=0$ (VSI), then $S$ is connected, and $\{0\}\subset S$ is the unique closed orbit in $S$. 
\item{} If $S$ consists of $n\geq 1$ connected components, $S_i$, then there exist unique closed orbits $Gx_i\subset S_i$, for each $i=1,...,n$. These closed orbits are necessarily Wick-rotated. 
\end{enumerate}
\end{prop}

\begin{proof}
First we note that since the function $\mathfrak{I}$ is polynomial, the set $S$ is closed in $\mathcal{V}$. Recall also that there is a finite number of closed orbits in $S$, hence, also in each component $S_i$. Furthermore, each non-closed orbit has a unique closed orbit on its boundary. 

Let $S_i$ be one of the connected components and assume that  $A_1$ and $A_2$ are two disjoint closed orbits in $S_i$; $A_1,A_2\subset S$. Let \[ U_I:=\bigcup \left\{Gx\subset S_i: A_I\cap \overline{Gx}\neq \emptyset \right\}, \quad I=1,2; \] 
i.e., the union of sets having $A_I$ as part of their closure. Consider the intersection $V=\overline{U_1}\cap\overline{U_2}$. There are now two possibilities: 
\par $V\neq \emptyset$: Since the intersection of two closed sets are closed, $V$ is nonempty and closed. Moreover, there are no orbits $Gx$ in $V$ which are closed since orbits in $U_I$ have a unique closed orbit on their boundary (namely $A_I$). Choose therefore a non-closed orbit $Gx$ in $V$. Then $\overline{Gx}$ contains a (unique) closed orbit, but since this is necessarily in $V$, this leads to a contradiction. 
\par $V=\emptyset$: Then $U_1$ and $U_2$ are disconnected. Define $W:=S_i\setminus U_1\cup U_2$ which is necessarily nonempty. Using the same argument as above, there needs to be  a non-closed orbit  in $W$ with a closed orbit $A_3\subset \overline{W}$ in its closure. Note that $A_3$ cannot be $A_1$ or $A_2$, because then the non-closed orbit in $W$ should have been in $U_I$ (hence, not in $W$). We can now do the same as above and define 
\[ U_3:=\bigcup \left\{Gx\subset S_i: A_3\cap \overline{Gx}\neq \emptyset \right\}.\] 
Then this again implies that there is another closed orbit $A_4$, etc. This must terminate since there is a finite number of closed orbits in $S_i$. Hence, this leads to a contradiction and the closed orbit in $S_i$ is thus unique. 

The first part of the proposition now follows since $\{0\}$ is obviously the only closed orbit in $S$ for which $\mathfrak{I}(S)=0$. 
\end{proof}

So in the sense of curvature tensors with identical invariants, each component $S_i$ is characterised by its unique closed orbit. Of course, this closed orbit could be $S_i$ itself (which it would be in the Riemannian case), but in the pseudo-Riemannian case more complicated structures of $S_i$ are possible. We should also recall that this is the  structure at a point $p\in M$. To study the structure in a neighbourhood of $M$ is a considerably more difficult task. 
\\

Consider now a Wick-rotation at $p$: $(M,g)\subset (M^{\mb{C}},g^{\mb{C}})\supset (\tilde{M},\tilde{g})$. Let $v\in \mathcal{V}$ and $\tilde{v}\in \tilde{\mathcal{V}}$ be the curvature tensors of covariant derivatives of the Riemann tensors (respectively) of $l$th order, i.e $\tilde{v}\in O(n,\mb{C})v$ are Wick-rotatable. Consider the function $\mathfrak{I}$ defined for $(M,g)$ as above. Then we can also consider (in exactly the same way as for $\mathcal{V}$ above) a function: $$\tilde{\mathfrak{I}}: ~\tilde{\mathcal{V}}\longrightarrow \R^{\tilde{N}},$$ for $\tilde{\mathcal{V}}$, and $\tilde{I}=\tilde{I}_{\tilde{k}}$ is generated by a finite set of generators of $\mathbb{R}[\tilde{\mathcal{V}}]^{{O(\tilde{p},\tilde{q})}}$ (restricted to the covariant derivatives of the Riemann tensor up to some $\tilde{k}$th order) w.r.t the action of $O(\tilde{p},\tilde{q})$ on $\tilde{\mathcal{V}}$. Let $1\leq l\leq max\{k,\tilde{k}\}$ where $k$ is as above. Define analogously: $$\tilde{S}:=\tilde{\mathfrak{I}}^{-1}(\tilde{\mathfrak{I}}(\tilde{v}))=\cup_{j=1}^m\tilde{S}_j,$$ where $\tilde{S}_j$ are the connected components. Set $G:=O(p,q), \ \tilde{G}:=O(\tilde{p},\tilde{q})$ and $G^{\mb{C}}:=O(n,\mb{C})$, and denote by the previous proposition $\tilde{G}\tilde{v}_j$ (respectively $Gv_i$) for the unique closed orbits in each component $\tilde{S}_j$(respectively $S_i$). 

By Section 6.1 and the notation there, we can choose $\Psi(-):=g\cdot -$ for some $g\in G^{\mb{C}}$ such that $Gv\sim g(\tilde{G})\cdot (g\cdot\tilde{v})$ are compatible real orbits, where $\rho^{\tilde{G}}_{\tilde{\mathcal{V}}}\cong \rho^{g(\tilde{G})}_{g(\tilde{\mathcal{V})}}$ as real representations via $g$. Now the map $\Psi$ is a morphism of affine complex varieties, and therefore the algebra of polynomial invariants: $\mathbb{R}[\tilde{\mathcal{V}}]^{\tilde{G}}\cong \mathbb{R}[g(\tilde{\mathcal{V}})]^{g(\tilde{G})}$ of the actions are related precisely via the action of $g$. Thus the set $\tilde{S}$ is mapped to $g\cdot \tilde{S}$, and $\tilde{I}$ is mapped to $g\cdot \tilde{I}$ and so on.

We have the following result: 

\begin{cor} Let $S$ and $\tilde{S}$ be defined as above, and $\tilde{v}\in O(n,\mb{C})v$ be as above. Then
\begin{enumerate}

\item{} For all $1\leq i\leq n$ and $1\leq j\leq m$ the tensors: $v_i$ and $\tilde{v}_j$ are Wick-rotatable tensors.
\item{} $\mathfrak{I}(S)=0\Leftrightarrow\tilde{\mathfrak{I}}(\tilde{S})=0$. In particular if $\mathfrak{I}(S)=0$ then $S\cap\tilde{S}\neq \emptyset$.
\item{} If $(\tilde{M},\tilde{g})$ is Riemannian, (i.e $\tilde{G}:=O(n)$ is compact), then there exist a $g\in O(n,\mb{C})$ such that $S_i\cap g\cdot \tilde{S}\neq \emptyset$ for all $1\leq i\leq n$. Thus for all $1\leq i\leq n$ there exist Cartan involutions $\theta_i$ of the metric $g(-,-)$ such that $\theta_i\cdot v_i=v_i$.

\end{enumerate} \end{cor}
\begin{proof} For all cases it is enough to assume $Gv\sim \tilde{G}\tilde{v}$ are compatible (see the paragraph before the statement). For case (1), suppose first that $Gv\subset V$ is closed, thus so is $\tilde{G}\tilde{v}\subset \tilde{\mathcal{V}}$. Hence since $v\in S_j$ for some $j$, then $Gv\subset S_j$ and is the unique closed orbit in $S_j$. Similarly $\tilde{G}\tilde{v}\subset \tilde{S}_i$ is the unique closed orbit for some $i$. So because $Gv\sim Gv_j$ for all $j$ and $\tilde{G}\tilde{v}\sim \tilde{G}\tilde{v}_i$ for all $i,j$ by the previous proposition, then also $Gv_j\sim \tilde{G}\tilde{v}_i$ for all $i,j$, and the closed case follows. Suppose now that $Gv$ is not closed. Then let $Gx\subset \overline{Gv}$ and $\tilde{G}\tilde{x}\subset \overline{\tilde{G}\tilde{v}}$ be the unique closed orbits in the closures. Now since $x$ and $v$ (respectively $\tilde{x}$ and $\tilde{v}$) have the same invariants then there are $i,j$ such that $Gx\subset S_i$ and $\tilde{G}\tilde{x}\subset \tilde{S}_j$. But since $Gv\sim \tilde{G}\tilde{v}$, then also $Gx\sim \tilde{G}\tilde{x}$ by uniqueness of closed orbits in the closure: $\overline{G^{\mb{C}}v}$, and so the statement follows.

For case (2), if $\mathcal{J}(S)=0$, then $0\in S$, and thus if $Gv_j\subset S_j$ is closed and $\tilde{G}\tilde{v}_i\subset \tilde{S}_i$ is closed, then by the proof of (1): $\tilde{G}\tilde{v}_i\sim Gv_j\sim G\cdot 0=\{0\}$, proving that $\mathcal{J}(\tilde{v}_i)=0$, and thus $\mathcal{J}(\tilde{S})=0$. The converse is symmetric so identical. The second statement follows since $0\in S\cap\tilde{S}$. 

For case (3), since $\tilde{S}=\tilde{G}\tilde{v}$, and $Gv_j\sim \tilde{G}\tilde{v}$ for all $j$ by following the proof of (1), then $Gv_j\cap \tilde{G}v\neq \emptyset$, i.e it follows that $Gv_j\cap \tilde{S}\neq\emptyset$, and so $S_j\cap \tilde{S}\neq \emptyset$. Now by Theorem \ref{t}, the last part of the statement follows.
 \end{proof}

\section*{Acknowledgements} 
This work was supported through the Research Council of Norway, Toppforsk
grant no. 250367: \emph{Pseudo-Riemannian Geometry and Polynomial Curvature Invariants:
Classification, Characterisation and Applications.}

\appendix


\begin{thebibliography}{abc}   


\bibitem{W2}
C. Helleland, S. Hervik, 
{\it J. Geom. Phys.}
{\bf 123} (2018) 343-361\\
 https://doi.org/10.1016/j.geomphys.2017.09.009.

\bibitem{W1}
C. Helleland, S. Hervik, 
{\it J. Geom. Phys.} {\bf 123} (2018) 424-429 \\
https://doi.org/10.1016/j.geomphys.2017.09.015.

\bibitem{VSI} V. Pravda, A. Pravdova, A. Coley and R. Milson, 
{\it Class. Quant. Grav. }
{\bf{19}}, 6213 (2002) [arxiv: gr-qc/0209024];  \\
A. Coley, R. Milson, V. Pravda, A. Pravdova, 
{\it Class. Quant. Grav.} {\bf{21}}, 5519 (2004);  \\ 
A. Coley, A. Fuster, S. Hervik and N. Pelavas, 
{\it Class. Quant. Grav. } 
{\bf{23}}, 7431 (2006). 

\bibitem{OP} S Hervik and A. Coley,  
{\it Class. Quant. Grav.} {\bf 27}, 095014  (2010) [arXiv:1002.0505];  \\
A. Coley and  S. Hervik, 
{\it  Class. Quant. Grav.} {\bf 27}, 
015002 (2009) [arXiv:0909.1160]. 


\bibitem{HC}
 S.~Hervik and A.~Coley,
  {\it Int.\ J.\ Geom.\ Meth.\ Mod.\ Phys.}\  {\bf 08}, 1679 (2011)\\
 S.~Hervik and A.~Coley,
 {\it  Class.\ Quant.\ Grav.}\  {\bf 28}, 015008 (2011)\\
 S.~Hervik,
  {\it Class.\ Quant.\ Grav.}\  {\bf 29}, 095011 (2012)

\bibitem{GW}
R. Goodman and N.R. Wallach, Symmetry, Representations and Invariants, Springer, 2009.

\bibitem{RS}
R.W. Richardson and P.J. Slodowy, {\it  J. London Math. Soc.} (2) {\bf 42}: 409-429 (1990).

\bibitem{PV}
V. Pessers and  J. Van der Veken, 
{\it J. Geom. Phys.}  {\bf 104} (2016) 163-174.
https://doi.org/10.1016/j.geomphys.2016.02.009, \\
\texttt{arXiv:1503.07354 [math.DG]}. 

\bibitem{minimal} 
  S.~Hervik, M.~Ortaggio and L.~Wylleman,
  {\it Class.\ Quant.\ Grav.}\  {\bf 30}, 165014 (2013)
  [arXiv:1203.3563 [gr-qc]]. 

\bibitem{CR}
Christoph Bohm, Ramiro A. Lafuente, \textit{Real geometric invariant theory}, arXiv:1701.00643 [math.DG].


\bibitem{EJ} 
P.Eberlein and M.Jablonski, \textit{Closed orbits of semisimple group actions and the real Hilbert-
Mumford function}, in {\it Contemp. Math.: New developments in Lie theory and Geometry}, vol. {\bf 491}, proceedings. AMS (2007). 

\bibitem{Neeb}
Joachim Hilgert and Karl-Hermann Nebb, \textit{Structure and Geometry of Lie groups}, Springer monographs in mathematics, 2012.


\bibitem{Mos55}
G. D. Mostow, 
{\it Ann. of Math. }(2) \textbf{62} (1955), 44-55.



\bibitem{Dong}
Dong Hoon Lee, \textit{The structure of complex Lie groups}, CHAPMAN and HALL/CRC Research Notes in Mathematics.

\bibitem{BC}
A.Borel and Harish-Chandra, 
{\it Ann. of Math} (2) 75 (1962) 485-535.

\bibitem{CGHP}
A. Coley, G.W. Gibbons, S. Hervik and C. N. Pope, 
{\it Class.\ Quant.\ Grav.} {\bf 25} 145017 (2008) [arXiv:0803.2438]


\bibitem{Bleecker}
D. D. Bleecker, 
{\it J. Di. Geo.} {\bf 14} 599--608 (1979)
\bibitem{universal} 
G.T. Horowitz and A. R. Steif,  
{\it Phys. Rev. Lett.} {\bf 64} 260--263 (1990)\\
S. Hervik, V. Pravda  and A. Pravdov\'a, 
{\it Class.\ Quant.\ Grav.} {\bf 31} 215005 (2014)  [arXiv:1311.0234]\\
S. Hervik, T. M\'alek, V. Pravda and A. Pravdov\'a,  
{\it Class.\ Quant.\ Grav.} {\bf 32} 245012 (2015) [arXiv:1503.08448]\\
S. Hervik , V. Pravda  and A. Pravdov\'a 
JHEP {\bf 10} 28 (2017) [arXiv:1707.00264]\\
S. Hervik and T. M\'alek  {\it Phys. Scr.} {\bf 93}, 085206 (2018) [arxiv: 1710.02164]



\end{thebibliography}
\end{document}